\documentclass{article}

\usepackage{PRIMEarxiv}
\usepackage{natbib}
\usepackage[utf8]{inputenc} 
\usepackage[T1]{fontenc}    

\usepackage{hyperref}       
\usepackage{url}            
\usepackage{booktabs}       
\usepackage{amsfonts}       
\usepackage{nicefrac}       
\usepackage{microtype}      
\usepackage{lipsum}
\usepackage{fancyhdr}       
\usepackage{graphicx}       
\usepackage{amsmath}
\usepackage{booktabs}
\usepackage{caption}
\usepackage{graphicx}
\usepackage{diagbox}
\usepackage{siunitx}
\usepackage{setspace}
\usepackage{subcaption}
\usepackage{amsthm}
\usepackage[table]{xcolor}
\usepackage{rotating}
\usepackage{xcolor}
\usepackage{footnote}
\usepackage{longtable}
\usepackage{enumitem} 
\usepackage{booktabs} 
\usepackage{array}
\usepackage[linesnumbered,lined,boxed,commentsnumbered]{algorithm2e} 
\usepackage{bbm}
\DeclareMathOperator*{\argmin}{arg\,min}
\DeclareMathOperator*{\argmax}{arg\,max} 

\newtheorem{remark}{Remark}
\newtheorem{theorem}{Theorem}

\newtheorem{corollary}{Corollary}

\theoremstyle{definition}

\definecolor{lightgray}{gray}{0.8}
\definecolor{darkgray}{gray}{0.9}
\definecolor{lgray}{rgb}{0.9, 0.9, 0.9}
\graphicspath{{media/}}     
\title{ Bivariate Variable Ranking for censored time-to-event data via  Copula Link Based Additive models
}

\author{
      Danilo Petti\thanks{Corresponding author.}\thanks{School of Mathematics, University of Essex, Wivenhoe Park, Colchester.\texttt{d.petti@essex.ac.uk}.}
       \and
       Marcella Niglio \thanks{Department of Statistics, University of Salerno, Via Giovanni Paolo II, 132, Fisciano. \texttt{mnigli@unisa.it}.}
       \and
       Marialuisa Restaino \thanks{Department of Statistics, University of Salerno, Via Giovanni Paolo II, 132, Fisciano. \texttt{mlrestaino@unisa.it}.}
}

\begin{document}
\maketitle

\begin{abstract}
In this paper, we present a variable ranking approach established on a novel measure to select important variables in bivariate Copula Link-Based Additive Models \citep{Marra2020}. The proposal allows for identifying two sets of relevant covariates for the two time-to-events without neglecting the dependency structure that may exist between the two survivals. The procedure suggested is evaluated via a simulation study and then is applied for analyzing the Age-Related Eye Disease Study dataset. The algorithm is implemented in a new \texttt{R} package, called \texttt{BRBVS}.
\end{abstract}


\clearpage
\tableofcontents

\section{Introduction}\label{intro}
In numerous studies of various scientific fields such as economics, genomics, sociology, biomedicine, and clinical trials, a large amount of variables is often monitored, sometimes exceeding even the number of observations ($p >> n$). Consequently, selecting the most relevant and significant variables becomes crucial to enhance the quality of estimation, prediction, and interpretation within models (among the others \cite{fan2001variable}).

During the last decades, an extensive body of literature on variable selection methods has played a pivotal role in nearly every field, especially in the context of linear models (see the reviews, \cite{fan2010overviewvarableselectio, desboulets2018review, HeinzeEtAl2018}). 

The traditional techniques widely applied for variable selection are, among others, forward selection, backward elimination, stepwise selection, and best-subset selection \citep{harrell2001}. Then, penalized variable selection methods based on the introduction of a penalty term in the likelihood function are largely studied in linear regression \citep{tibshirani1996regression}, in generalized linear models (GLM) \citep{Friedman2010}, accelerated failure time models (AFT) \citep{ park2018aft}, Cox's proportional hazards and frailty models \citep{tibshirani1997lasso, fan2002variable, zou2008note, fan2001variable}, in copula survival models \citep{Kwon2019copula}, and for multivariate survival data \citep{cai2005variable}.

However, all the penalized methods may not perform well when the dimension of data is high and ultra-high, leading to some problems related to computational expediency, statistical accuracy, and algorithmic stability. A proposal to overcome these difficulties is given by the screening procedures (see, e.g. \cite{Fan2009VS}) that identify the relevant covariates evaluating their importance, measured in terms of the association between the dependent variable and each feature. 

The screening procedure based on this idea is the Sure Independence Screening (SIS), originally proposed by \cite{fan2008} in ordinary linear regression, where the importance of variables is measured by estimating the marginal correlation between each covariate and the response variable. Subsequently, after deleting the uninformative variables, standard penalized variable selection methods (e.g. LASSO, SCAD, Adaptive LASSO, etc.) can be applied to the reduced dataset. 

SIS technique is extended to other models and several measures for variable ranking are proposed (i.e.\ generalized linear models \citep{FanSong2010}, Cox proportional hazards and Accelerator Failure models \citep{fan2010high, Khan2022, song2014censored, wang2022spearman, li2016survival} and more complex models \citep{Hall2009, Fryzlewicz2012, Shao2014MartingaleDC, li2012b, he2013quantile}).

In the regression domain, \cite{Baranowski2020} proposed 
a procedure, called Ranking-Based Variable Selection (RBVS), that simultaneously identifies subsets of covariates that consistently appear to be important across subsamples extracted from the data.

To the best of our knowledge, there is no extensive literature on screening and variable selection methods for the bivariate copula survival models in the presence of censoring and even less for high dimensional datasets, i.e.\ when the number of covariates is higher than the number of observations.

However, the procedures mentioned above, proposed in the literature for regression models, cannot be automatically adapted to the domain of bivariate survival copulas \citep{Marra2020, PETTI22}, 
for two main reasons: i) they do not produce unique sets of important variables for each time-to-event;  ii) they mainly focus on assessing linear relationships between times and covariates, neglecting the structure of the copula, which is essential in the bivariate survival data context.

Therefore, our contribution aims to fill this gap with two aims. First, we propose an approach able to identify the relevant covariates for bivariate copula survival data by considering marginally the contribution of each feature (as described in the following sections). For this reason, it can also be easily adopted in high-dimensional settings.  In particular, we extend the algorithm proposed in \cite{Baranowski2020}  to the bivariate survival models where the dependence is measured by the copula parameter, and the ranking of the covariates is marginally made for both events of interest. Therefore, since different features may influence the times-to-event, the proposed procedure might capture the differences in the effects of all covariates on each event without neglecting the dependence structure captured via the copula function.

Second, we propose a measure to evaluate the importance of variables and consequently rank them. It is based on the Fisher information matrix, which can handle the probabilistic copula structure between the two survival variables, and it also has an easy geometrical interpretation, overcoming the drawbacks of other measures available in the literature \citep[e.g.,][among others]{zhu2011model,he2013quantile,song2014censored,wang2022spearman}.

The proposed procedure is implemented from scratch and in a new \texttt{R} package called \texttt{BRBVS}, hence enabling any user to carry out the analysis and generate easy-to-interpret numerical and visual summaries.

The paper is organized as follows. In Section \ref{Ch:modelformulation}, an overview of the  Copula link-based additive survival model is provided. In Section \ref{Ch:BRBVS}, the bivariate variable selection algorithm is presented, and in particular, a new variable ranking measure is introduced; in Section \ref{Ch:Main_simulation}, the main results of the simulation study are presented;  Section \ref{Ch:Application} shows the empirical performance of the variable selection algorithm in the presence of real data. Section \ref{ch:conclusion} concludes the paper with a discussion. The online Supplementary Material includes: i) details about the theoretical properties of the measure introduced and the consistency of the bivariate variable selection procedure; ii) further results from the simulation study; and iii) a comprehensive discussion about the examined data and the \texttt{R} implementation of the BRBVS procedure.

\section{Theoretical Framework}
\label{Ch:modelformulation}
In this section, we will recall the theoretical framework of the Copula Link-Based Survival Additive model. All technical details and proofs can be found in \cite{Marra2020}, where the model is originally proposed and then extended to all censoring scenarios in \cite{PETTI22}.

Let $\left(T_{1i}, T_{2i}\right)$ and $\left(C_{1 i}, C_{2 i}\right)$ be a pair of survival times and the censoring times respectively, for unit $i$, with $i=1,\dots,n$, where $n$ is the sample size and $\left(C_{1 i}, C_{2 i}\right)$ are assumed to be independent of $\left(T_{1i}, T_{2i}\right)$.

Let $\mathbf{x}_i$ be the $i^{\text{th}}$ row-vector of the design matrix $\mathbf{X}$ having  dimension $(n \times p)$, where $p$ is the number of covariates. 

The conditional marginal survival functions for $T_{\nu i}$ and the conditional joint survival function are, respectively, given by \[S_{\nu}(t_{\nu i}  \mid \mathbf{x}_{\nu i}; \boldsymbol{\beta}_{\nu})=P\left(T_{\nu i}>t_{\nu i} \mid \mathbf{x}_{\nu i}; \boldsymbol{\beta}_{\nu}\right)\]
for $\nu=1,2$, and $S\left(t_{1 i}, t_{2 i} \mid \mathbf{x}_{i}; \boldsymbol{\delta}\right)=P\left(T_{1 i}>t_{1 i}, T_{2 i}>t_{2 i} \mid \mathbf{x}_{i}; \boldsymbol{\delta}\right)$.

Then, the marginal survivals for the observed time values $t_{1i}$ and $t_{2i}$  are linked by a copula function:
\begin{align}
\label{eq:jointSurvival}
S\left(t_{1 i}, t_{2 i} \mid \mathbf{x}_{i}; \boldsymbol{\delta}\right)=C\left(S_{1}\left(t_{1 i} \mid \mathbf{x}_{1 i}; \boldsymbol{\beta}_{1}\right), S_{2}\left(t_{2 i} \mid \mathbf{x}_{2 i}; \boldsymbol{\beta}_{2}\right); m\left\{\eta_{3 i}\left(\mathbf{x}_{3 i}; \boldsymbol{\beta}_{3}\right)\right\}\right),
\end{align}
where 
\begin{itemize}
    \item $\boldsymbol{\delta}^{\top}=(\boldsymbol{\beta}_1^{\top}, \boldsymbol{\beta}_2^{\top}, \boldsymbol{\beta}_3^{\top}) \in \mathbb{R}^W$ is the vector of parameters  such that $\boldsymbol{\beta}_1 \in \mathbb{R}^{W_1}$, $\boldsymbol{\beta}_2\in \mathbb{R}^{W_2}$, $\boldsymbol{\beta}_3\in \mathbb{R}^{W_3}$, with $W_{\nu} \subseteq p$ and $W=\sum_{\nu}W_{\nu} \ge p$, for $\nu = 1, 2, 3$.
    \item $\mathbf{x}_i^{\top}=(\mathbf{x}_{1i}^{\top}, \mathbf{x}_{2i}^{\top}, \mathbf{x}_{3i}^{\top})\in \mathbb{R}^W$  the vectors of covariates such that $\mathbf{x}_{1i}\in \mathbb{R}^{W_1}$, $\mathbf{x}_{2i}\in \mathbb{R}^{W_2}$, $\mathbf{x}_{3i}\in \mathbb{R}^{W_3}$ can be sub-vectors of (or be equal to) $\mathbf{x}_i$. 
    \item $C(\cdot):(0,1)^2\to (0,1)$ is the copula function, with coefficient \\ $m\{\eta_{3 i}$ $(\mathbf{x}_{3 i} ; $$\boldsymbol{\beta}_{3})\}$ that captures the possibly varying conditional dependence of $\left(T_{1 i}, T_{2 i}\right)$ across observations \citep[e.g.,][]{Marra2020, pattonT, Sklar}.
    \item  $\eta_{3 i}\left(\mathbf{x}_{3 i} ; \boldsymbol{\beta}_{3}\right) \in \mathbb{R}$ is a predictor which includes generic additive covariate effects.
    \item $m(\cdot)$ is an inverse monotonic and differentiable link function which ensures that the dependence parameter lies in a proper range. 
\end{itemize}

The marginal survivals $S_1(\cdot)$ and $S_2(\cdot)$ are modeled by generalized survival or link-based models \citep{Liu2018, Royston2002}, leading to 
$S_{\nu}\left(t_{\nu i} \mid \mathbf{x}_{\nu i}; \boldsymbol{\beta}_{\nu}\right)=G_{\nu}\left\{\eta_{\nu i}\left(t_{\nu i}, \mathbf{x}_{\nu i}; \boldsymbol{\beta}_{\nu}\right)\right\}$, where $G_\nu(\cdot)$ is the inverse link function of $g(\cdot):[0,1]\to \mathbb{R}$ and $\eta_{\nu i}\left(t_{\nu i}, \mathbf{x}_{\nu i}; \boldsymbol{\beta}_{\nu}\right) \in \mathbb{R}$, for $\nu=1,2$, are the additive predictors that must include baseline functions of time (or a stratified set of functions of time) as clarified in \cite{Marra2020}.

The main difference between $\eta_{\nu i}\left(t_{\nu i}, \mathbf{x}_{\nu i}; \boldsymbol{\beta}_{\nu}\right)$ for $\nu=1,2$ and $\eta_{3 i}\left(\mathbf{x}_{3 i}; \boldsymbol{\beta}_{3}\right)$ is that in $\eta_{\nu i}$ two smooth functions of time are included, and $t_{\nu i}$ can be treated as regressors. Dropping the dependence on the covariates, a vector $\mathbf{z}_{\nu i}$ for $\nu =1,2,3$ is introduced, where for $\nu=1,2$,  $\mathbf{z}_{\nu i}=(\mathbf{x}_{\nu i}, t_{\nu i})$, and for $\nu=3$, $\mathbf{z}_{3 i}=\mathbf{x}_{3 i}$. 
Thus, the three additive predictors can be written as:
\begin{equation}
\label{eq:AdditivePredictor}
\eta_{\nu i}=\beta_{\nu 0}+\sum_{k_{\nu}=1}^{K_{\nu}} s_{\nu k_{\nu}}\left(\mathbf{z}_{\nu k_{\nu} i}\right), \quad i=1, \ldots, n; \quad \nu=1,2,3,
\end{equation}
where $\beta_{\nu 0} \in \mathbb{R}$ denotes an overall intercept, $\mathbf{z}_{\nu k_{\nu} i}$ is the $k_{\nu}^{t h}$ sub-vector of the complete vector $\mathbf{z}_{\nu i}$ and the $K_{\nu}$ functions $s_{\nu k_{\nu}}\left(\mathbf{z}_{\nu k_{\nu}}\right.$) represent generic effects which are chosen according to the type of covariate(s). As clarified in \cite{Marra2020}, these functions can be expressed as a linear combination of basis functions $\boldsymbol{b}_{\nu k_{\nu}}(\mathbf{z}_{\nu k_{\nu} i})=(b_{\nu k_{\nu} 1}, \dots, b_{\nu k_{\nu} J_{\nu k_{\nu}}})^{\top} \in \mathbb{R}^{J_{\nu k_\nu}}$  and  the vector of regression coefficients as $\boldsymbol{\beta}_{\nu k_{\nu}}=\left(\beta_{\nu k_{\nu} 1}, \ldots, \beta_{\nu k_{\nu} J_{\nu k_{\nu}}}\right)^{\top} \in \mathbb{R}^{J_{\nu k_{\nu}}}$ such that $\mathbf{s}_{\nu k_{\nu}} (\mathbf{z}_{\nu k_{\nu } i })=\boldsymbol{b}_{\nu k_{\nu}}(\mathbf{z}_{\nu k_{\nu}})^{\top} \boldsymbol{\beta}_{\nu k_{\nu}} $. The aforementioned formulation allows for taking into account a wide range of covariate effects, including linear, nonlinear, and spatial effects \citep{Wood2017}. 
Additionally, this framework incorporates various copula functions and link functions.
For a detailed overview, see Table 1 in \cite{Marra2020}.

Due to the presence of censoring, $T_{\nu i}$ are observed only if $T_{\nu i}\in(L_{\nu i}, R_{\nu i})$, where $L_{\nu i}$ and $R_{\nu i}$ denote the left- and right- censored times, respectively (for $\nu=1,2$). Given that there are different potential combinations of censoring scenarios to be accounted for, two indicator functions  $\gamma_{C}$ and  $\gamma_{U}$ are introduced. Thus,  $\gamma_{C}$ takes a value of 1 if the observation is interval/right/left-censored, and 0 otherwise; $\gamma_{U}$ assumes a value of 1 if the observation is uncensored, and 0 otherwise. More details are in  \cite{PETTI22}.

The log-likelihood is given by:
\begin{align}
\label{eq:lik}
\ell(\boldsymbol \delta) & = \gamma_{U_{1i}} \gamma_{U_{2i}}\sum_{i=1}^n \log [\nabla^2_{t_{1i} t_{2i}}  C_{t_{1i} t_{2i}}] + \gamma_{C_{1i}} \gamma_{C_{2i}} \sum_{i=1}^n \log [F(l_{1i}, l_{2i}, r_{1i}, r_{2i})] + \nonumber \\
& \quad + \gamma_{U_{1i}} \gamma_{C_{1i}} \sum_{i=1}^n \log [\nabla_{t_{1i}}(C_{t_{1i} r_{2i}}-C_{t_{1i} l_{2i}})] + \nonumber\\
&\quad  +\gamma_{C_{1i}} \gamma_{U_{1i}} \sum_{i=1}^n \log [\nabla_{t_{2i}} (C_{r_{1i} t_{2i}}- C_{l_{1i} t_{2i}})],
\end{align}
where
\[\nabla^2_{t_{1i} t_{2i}}  C_{t_{1i} t_{2i}} = \frac{\partial^2}{\partial t_{1i} \partial t_{2i}}  C\left\{G_1(\eta_{1i}(t_{1i})),G_2(\eta_{2i}(t_{2i}));\theta_i \right\}, \]
\begin{align*}
\nabla_{t_{1i}}(C_{t_{1i} r_{2i}}-C_{t_{1i} l_{2i}}) &= \frac{\partial}{\partial t_{1i}} \bigg(C \lbrace G_1(\eta_{1i}(t_{1i}) ), G_2(\eta_{2i}(r_{2i}) );\theta_i \rbrace+ \\
&\quad -  C \lbrace G_1(\eta_{1i}(t_{1i}) ), G_2(\eta_{2i}(l_{2i}) );\theta_i \rbrace \bigg), 
\end{align*}
\begin{align*}
\nabla_{t_{2i}} (C_{r_{1i} t_{2i}}- C_{l_{1i} t_{2i}}) &= \frac{\partial}{\partial t_{2i}} \bigg(C \lbrace G_1(\eta_{1i}(r_{1i}) ), G_2(\eta_{2i}(t_{2i}) );\theta_i \rbrace +\\
&\quad -  C \lbrace G_1(\eta_{1i}(l_{1i}) ), G_2(\eta_{2i}(t_{2i}) );\theta_i\rbrace \bigg),
\end{align*}
\begin{align*}
F(l_{1i}, l_{2i}, r_{1i}, r_{2i}) &= C\lbrace G_1(\eta_{1i}(l_{1i})), G_2(\eta_{2i}(l_{2i})); \theta_i \rbrace +\\
&\quad - C\lbrace G_1(\eta_{1i}(l_{1i})), G_2(\eta_{2i}(r_{2i}));\theta_i \rbrace + \\
&\quad - C\lbrace G_1(\eta_{1i}(r_{1i})), G_2(\eta_{2i}(l_{2i}));\theta_i \rbrace + \\
&\quad + C\lbrace G_1(\eta_{1i}(r_{1i})), G_2(\eta_{2i}(r_{2i}));\theta_i\rbrace. 
\end{align*}
Given some difficulties in the estimation \citep[e.g.,][]{Rupp03}, 
a penalized version of the log-likelihood is maximized
\begin{equation}
\label{eq:penLik}
\ell_{p}(\boldsymbol{\delta})=\ell(\boldsymbol{\delta})-\frac{1}{2} \boldsymbol{\delta}^{\top} \mathbf{S} \boldsymbol{\delta},
\end{equation}
where $\ell_{p}$ is the penalized log-likelihood,  $\mathbf{S}=\operatorname{diag}\left(\mathbf{D}_{1}, \mathbf{D}_{2}, \mathbf{D}_{3}\right)$, with $\mathbf{D}_{\nu}=\operatorname{diag}(0, \lambda_{\nu 1}  \mathbf{D}_{\nu 1},  \ldots, $ $\lambda_{\nu K_{\nu}} \mathbf{D}_{\nu K \nu})$ the overall penalties, 
and \\ $\lambda_{\nu k_{\nu}}= $$(\lambda_{\nu 1}, \ldots,   \lambda_{\nu K_{\nu}})^{\top}   \in[0, \infty)$ the smoothing parameters which control the trade-off between fit and smoothness. 

All technical details and proofs are available in \cite{Marra2020} for the right censored data and in \cite{PETTI22} for a more general censoring setting.

\section{Bivariate Ranking-Based Variable Selection } \label{Ch:BRBVS}

As discussed in Section \ref{intro}, the variable screening and selection are mainly developed in the presence of regression models with the univariate response, whereas their extension to multiple response data is considered only in well-defined cases (see e.g.\ \cite{he2021surescreening}). We here extend the Ranking-Based Variable Selection (RBVS) algorithm of \cite{Baranowski2020} to model \eqref{eq:jointSurvival} by considering two main steps: the first devoted to the variable ranking and the second to the variable selection. 

Since the RBVS approach is here generalized to the model \eqref{eq:jointSurvival}, it will be called the Bivariate RBVS algorithm (BRBVS). 

As clarified in \cite{Marra2020}, the model \eqref{eq:jointSurvival} has a common design matrix for all $\eta_{\nu}$ ($\nu=1, 2, 3$), and the vectors ${\bf x}_\nu$ can include covariates that might be (or not) the same for all $\nu$. Moreover, the relevant variables cannot be separately selected by the two margins $\eta_{\nu}$ ($\nu=1,2$) for two main reasons:  i) this would lead to completely ignoring the dependency structure between the two time-to-event. ii) Including the covariates in $\eta_3$ makes the model \eqref{eq:jointSurvival} even more complex, as this could introduce spurious correlations.

For all these reasons, the \cite{Baranowski2020} algorithm cannot be applied directly and needs to be reformulated. In more detail, the algorithm is generalized to the bivariate survival domain in Section \ref{sec:BRBVS}; a new measure to rank the covariates and an estimation setup to control the computational burden are introduced in Section \ref{SubCh:FisherMetric}.

\subsection{Bivariate RBVS algorithm (BRBVS)} \label{sec:BRBVS}

Let $\mathcal A_\nu\subset (1,2,\ldots, p)$ be the indices that identify a subset of covariates included in the design matrix $\bf X$, and let $|\mathcal A_\nu|=k_{\nu}$ the cardinality of the covariates included in the first and second margins, $\eta_\nu$, $\nu=1,2$. Furthermore, let $\mathcal R_{\nu j}({\bf z})$ be the ranking of the covariate $j^{\text{th}}$ for $j=1,2,\ldots, p$, in $\eta_\nu$, based on a measure $\hat\omega_{\nu j}({\bf z})$, with ${\bf z}=({\bf z}_1, \ldots, {\bf z}_n)$ where ${\bf z}_i=(t_{i1}, t_{i2}, x_{i1}, \ldots, x_{ip})$. The measure $\hat\omega_{\nu j}$ is jointly defined for both $\eta_\nu$, $\nu=1,2$, to assess the importance of each covariate and is such that $\hat\omega_{\nu j}>\hat\omega_{\nu (j+1)}$, for $j=1,\ldots, (p-1)$.

Under the theoretical framework in Section \ref{Ch:modelformulation}, we can define, for each $\eta_\nu$, a unique top-ranked set of covariates such that the probability that $k_\nu$ variables are included in $\eta_\nu$ is:
\begin{equation}\label{eq:pi_v}
\pi_\nu(\mathcal A_\nu)=P(\left\{\mathcal R_{\nu_1}({\bf z}), \mathcal R_{\nu_2}({\bf z}),\ldots, \mathcal R_{\nu_{k_\nu}}({\bf z})\right\}=\mathcal A_\nu),
\end{equation}
with $\pi_\nu(\mathcal A_\nu)=1$, if $\mathcal A_\nu=\emptyset$.

To estimate the probability $\pi_\nu(\mathcal A_\nu)$, we follow the same bootstrap approach in \cite{Baranowski2020}: we consider $B$ bootstrap replicates from $\bf{X} $ and for each of them we draw uniformly without replacement $r=\lfloor n/m \rfloor$ random samples of size $m$ from the dataset (with $\lfloor n/m \rfloor$ the integer part of the ratio). Then, the estimate of the probability $\pi_\nu(\mathcal A_\nu)$ for the sets of covariates $\mathcal A_\nu$ is given by:

\begin{equation}\label{eq:hat_pi_v}
\hat\pi_{\nu,m}(\mathcal A_\nu)=B^{-1}\sum_{b=1}^Br^{-1}\sum_{q=1}^r{\bf 1}(\mathcal A_\nu|I_{b_q}),
\end{equation}
where ${\bf 1}(\cdot)$ is an indicator function which assumes value 1 when the covariates indexed in $\mathcal A_\nu$ are top-ranked and 0 otherwise; $I_{b_q}$ is the $q^{\text{th}}$ subsample extracted from the data, for $q=1,\ldots, r$, $\mathcal A_\nu|I_{b_j}$ is the subset of $k_\nu$ covariates of $\eta_\nu$, with $k_\nu=0,1,\ldots, p-1$, whose ranking is computed by $I_{b_q}$.

These probabilities allow us to define the top-ranked covariates for the two margins $\eta_\nu$:
\begin{equation}\label{eq:max_pi}
\hat{\mathcal A}_{\nu,m,k_v}=\underset{\mathcal A_\nu\in \Omega_{k_\nu}}{\argmax}\quad\hat\pi_{\nu,m}(\mathcal A_\nu),
\end{equation}
with $\Omega_{k_\nu}$ the set of all permutations of $\{1,\ldots, k_\nu\}$.

In practice, the variables included in $\hat{\mathcal A}_{\nu,m,k_\nu}$ are those with maximum probability or, in other words, the set of variables that most often are at the top of the ranking in the $B\cdot r$ sub-samples.

The ranking of covariates is widely used in the screening literature (see, among others, \citep{fan2008}) and, in most cases, the definition of a threshold is required for data reduction purposes. This definition is related to a prespecified value for the number of covariates (as in \cite{fan2008})  or to evaluations that considering the magnitude of the probabilities \eqref{eq:hat_pi_v} allows to identify the set of relevant covariates.

In the presence of the model \eqref{eq:jointSurvival}, the definition of a threshold has additional difficulties due to the presence of two subsets of covariates to select, and in this case, the application of the same threshold value may not be an appropriate choice. In fact, the order of magnitude of the probabilities \eqref{eq:hat_pi_v} should be quite different for the two margins, and then the definition of a threshold can lead to over- or under-selecting the variables. By contrast, even the definition of two prespecified thresholds can lead to a subjective variable selection.

For these reasons, following the idea in \cite{Baranowski2020}, for each $\eta_\nu$ we select the subset $\hat s_\nu$ of important variables such that:
\begin{equation}\label{eq:hat_s_v}
\hat{s}_\nu=\underset{k_\nu=0, \ldots, k_{\max}-1}{\argmin}\frac{[\hat\pi_{\nu,m}(\hat{\mathcal A}_{\nu, m, k_\nu+1})]^\tau}{\hat\pi_{\nu,m}(\hat{\mathcal A}_{\nu, m, k_\nu})},
\end{equation}
with $\tau\in(0,1]$ and $k_{\max}$ the maximum number of covariates considered to compute the ratio \eqref{eq:hat_s_v}, for $\nu=1,2$.

In practice, for each $\eta_\nu$ and given $\tau$, the number of selected variables is chosen by looking at the ratio in \eqref{eq:hat_s_v} and then $k_\nu$ variables are included in $\hat{s}_\nu$ such that $[\hat\pi_{\nu,m}(\hat{\mathcal A}_{\nu, m, k_\nu+1})]^\tau/\hat\pi_{\nu,m}(\hat{\mathcal A}_{\nu, m, k_\nu})$ drastically decreases concerning $[\hat\pi_{\nu,m}(\hat{\mathcal A}_{\nu, m, k_\nu})]^\tau/\hat\pi_{\nu,m}(\hat{\mathcal A}_{\nu, m, k_\nu-1})$.

In the Supplementary Material \ref{ch:teoretProp}, the unicity of the sets of important variables for the two margins (Corollary \ref{th:firstcorollary}) and the consistency of the variable selection procedure (Corollary \ref{th:Secondcorollary}) are shown. Furthermore, in \ref{ch:Rsoftware}, the main steps of the Algorithm are outlined.

Unlike \cite{Baranowski2020}, which deals mainly with linear regression models, the most challenging step in our model setup is the second, where a proper metric must be defined.

In the univariate survival domain, the Mutual Information (MI), also known as Copula Entropy, should be an effective tool for identifying significant variables either in regression \citep{cheng2022variable, frenay2013mutual}, and survival \citep{ma2022copula} contexts.

Its utility stems from three main characteristics: i) MI's applicability transcends variable types, extending beyond mere quantitative variables; ii) it adeptly captures various relationships among random variables, encompassing both linear and non-linear associations; iii) MI is robust against data transformations, including translations and rotations, and it ensures consistent analysis outcomes. These attributes should make MI a metric aligned with our research goals, even if its use has some limitations.

A notable challenge is its inability to be consistently expressed in a closed form, leading to computational intensity in practical applications. To empirically evaluate MI's efficacy in variable selection, we conduct a Monte Carlo simulation (Supplementary Material  \ref{ch:simMIvsFIM}) based on the model framework outlined in Section \ref{Ch:modelformulation}. The findings reveal limitations in MI's effectiveness for variable selection in our domain, leading us to introduce a new measure based on the Fisher Information (Section \ref{SubCh:FisherMetric}).  Our proposal seeks to determine the relevance of the covariates, exploiting the information of the joint estimation via the log-likelihood function. This is achieved by deriving a measure capable of simultaneously evaluating the coefficient estimates and the copula structure involved in the estimation process. Such a development could offer a more holistic and computationally efficient tool for variable selection also in many other domains.
\subsection{Measure Proposal and its estimation} \label{SubCh:FisherMetric}
Given the limitations of MI discussed in Section \ref{sec:BRBVS} and empirically evaluated in Supplementary Material \ref{ch:simMIvsFIM}, we introduce a new measure for the class of model(s) recalled in Section \ref{Ch:modelformulation}. The ideal measure in this domain should be able to: i) take into account the copula function in the selection process; ii) evaluate marginally the contribution of each covariate on $T_{\nu}$ $(\nu=1,2)$. To justify our proposed measure, we need to give some details on the goodness of fit.

It can be easily proved that the Akaike information criterion (AIC) \citep{Akaike} and the  Bayesian information criterion (BIC)  for the class of model(s) discussed in Section \ref{Ch:modelformulation} are respectively  $AIC = 2edf+ ||\boldsymbol{M}- \sqrt{-\boldsymbol{H}} \boldsymbol{\hat{\delta}}||^2$ and $BIC = edf \log{(n)} + ||\boldsymbol{M}- \sqrt{-\boldsymbol{H}} \boldsymbol{\hat{\delta}}||^2$ \citep[see][]{Marra2020}, where $edf$ are the effective degrees of freedom, $\|\cdot\|$  is the Eucledian norm, $\mathbf{M}=\boldsymbol{\mu}_{\mathbf{M}}+\boldsymbol{\epsilon}$, $\boldsymbol{\mu}_{\mathbf{M}}=\sqrt{-\mathbf{H}} \boldsymbol{\delta}^0$, $\boldsymbol{\epsilon}=\sqrt{-\mathbf{H}}^{-1} \mathbf{g}$,  $\mathbf{g}^{\top}=(\mathbf{g}_1^{\top}, \mathbf{g}_2^{\top}, \mathbf{g}_3^{\top})$ is the gradient vector of $\ell(\boldsymbol\delta)$ defined in \eqref{eq:lik}, such that $\mathbf{g}_1=\frac{\partial \ell(\boldsymbol{\delta}) }{\partial  \boldsymbol{\beta}_1}\Bigr|_{\substack{\boldsymbol{\beta}_1=\hat{\boldsymbol{\beta}}_1}}$,$\mathbf{g}_2=\frac{\partial \ell(\boldsymbol{\delta}) }{\partial  \boldsymbol{\beta}_2}\Bigr|_{\substack{\boldsymbol{\beta}_2=\hat{\boldsymbol{\beta}}_2}}$, $\mathbf{g}_3=\frac{\partial \ell(\boldsymbol{\delta}) }{\partial  \boldsymbol{\beta}_3}\Bigr|_{\substack{\boldsymbol{\beta}_3=\hat{\boldsymbol{\beta}}_3}}$ , $\mathbf{H}$ is the Hessian matrix of $\ell(\boldsymbol\delta)$ whose elements are $\mathbf{H}_{\nu \nu'}=\frac{\partial^2 \ell(\boldsymbol{\delta}) }{\partial  \boldsymbol{\beta}_\nu \partial \boldsymbol{\beta}_{\nu'}}\Bigr|_{\substack{\boldsymbol{\beta}_\nu=\hat{\boldsymbol{\beta}}_\nu}\\ \substack{\boldsymbol{\beta}_{\nu'}=\hat{\boldsymbol{\beta}}_{\nu'}}}$ $\nu,{\nu'}=1,2,3$, $\hat{\boldsymbol{\delta}}^{\top}=(\hat{\boldsymbol{\beta}_1}^{\top},\hat{\boldsymbol{\beta}_2}^{\top}, \hat{\boldsymbol{\beta}_3}^{\top})$ 
is the estimated parameter vector, $\boldsymbol{\delta}^0$ is the vector of parameters that minimizes the Kullback-Leibler distance between the theoretical and empirical likelihood.  Given the effective degrees of freedom (whose meaning will be clarified in Remarks \ref{rk:edfinCopulaModel} and \ref{rk:edfinBRBVS}), we have 
\begin{align*}
\mathbb{E}[|| \boldsymbol{M}- \sqrt{-\mathbf{H}} \hat{\boldsymbol{\delta}}||^2]=&
\mathbb{E}[|| \sqrt{-\mathbf{H}} \boldsymbol \delta^0+\sqrt{-\mathbf{H}}^{-1}\mathbf{g}- \sqrt{-\mathbf{H}} \hat{\boldsymbol{\delta}}||^2]\\
=&\mathbb{E}[
(\sqrt{-\mathbf{H}}\boldsymbol \delta^0+\sqrt{-\mathbf{H}}^{-1}\mathbf{g})^{\top} (\sqrt{-\mathbf{H}}\boldsymbol \delta^0+\sqrt{-\mathbf{H}}^{-1}\mathbf{g})+\\
&-(\sqrt{-\mathbf{H}}\boldsymbol \delta^0+\sqrt{-\mathbf{H}}^{-1}\mathbf{g})\sqrt{-\mathbf{H}} \hat{\boldsymbol{\delta}}+
\\
&-\hat{\boldsymbol{\delta}}^{\top}\sqrt{-\mathbf{H}} (\sqrt{-\mathbf{H}}\boldsymbol \delta^0+\sqrt{-\mathbf{H}}^{-1}\mathbf{g})
+\hat{\boldsymbol{\delta}}^{\top}(-\mathbf{H}) \hat{\boldsymbol{\delta}}]\\
=&\mathbb{E}[
\boldsymbol{\delta}^{0\top} (-\mathbf{H})\boldsymbol{\delta}^{0} + \boldsymbol{\delta}^{0\top} \mathbf{g}+  \mathbf{g}^{\top} \boldsymbol{\delta}^{0}+ \mathbf{g}^{\top} \mathbf{H}^{-1}\mathbf{g}+\\
&- (\boldsymbol{\delta}^{0\top} (-\mathbf{H}) \hat{\boldsymbol{\delta}}+ \mathbf{g}^{\top} \hat{\boldsymbol{\delta}} ) -( \hat{\boldsymbol{\delta}^{\top}} (-\mathbf{H}) \boldsymbol{\delta}^0+ \boldsymbol{\delta}^{\top} \boldsymbol{g})+ \hat{\boldsymbol{\delta}}^{\top} (-\mathbf{H}) \hat{\boldsymbol{\delta}}]\\
=&||\sqrt{-\mathbf{H}}\boldsymbol \delta^0-\sqrt{-\mathbf{H}} \hat{\boldsymbol{\delta}}||^2.
\end{align*}
The last term is a $\ell_2$ norm between two $W$ dimensional points adjusted for the information contained in a given sample, unless in some degenerative cases, the quantity $(\boldsymbol{\delta}^0-\hat{\boldsymbol{\delta}})$ is unknown, hence impossible to evaluate. 
However, using the plug-in method, we can evaluate $||\sqrt{-\mathbf{H}} \hat{\boldsymbol{\delta}}||^2$ with respect to $-\mathbb{E}(\mathbf{H}) = \boldsymbol{\mathcal{I(\delta)}}$ (with $\boldsymbol{\mathcal{I(\delta)}}$ the Fisher information matrix).
Since the accuracy of the estimates is measured by the sharpness of the likelihood function, where its curvature is captured by the main diagonal of $\mathbf{H}$, the elements on the main diagonal of the Hessian can be seen as directional derivatives with respect to the vectors of the canonical basis.

Let $\mathbf{e}^{(\nu)}_{j}\in \mathbb{R}^{W_\nu}$ for $\nu=1,2$, be a column vectors with zero everywhere except in the $j^{\text{th}}$ position and $\boldsymbol{\mathbf{\beta}}_\nu$ be the coefficient vector associated to the $\nu^{\text{th}}$ margin, such that  $\boldsymbol{\delta}^{\top}=(\boldsymbol{\mathbf{\beta}}_1,\boldsymbol{\mathbf{\beta}}_2,\boldsymbol{\mathbf{\beta}}_3)^{\top}\in \mathbb{R}^W$ where  $W=W_1+W_2+W_3$. 
Recalling that \(\boldsymbol{\mathcal{I}}(\boldsymbol{\delta})\) is the Fisher information matrix, we can denote its diagonal elements as  \(\bigoplus_{\nu=1}^{3} \boldsymbol{\mathcal{I}}_{\nu\nu}(\boldsymbol{\beta}_\nu) = \text{diag}\{\boldsymbol{\mathcal{I}}_{11}(\boldsymbol{\beta}_1), \boldsymbol{\mathcal{I}}_{22}(\boldsymbol{\beta}_2), \boldsymbol{\mathcal{I}}_{33}(\boldsymbol{\beta}_3)\}\) where \(\bigoplus\) represents the direct sum operation. After estimating the vector of parameters, and noticing that \(\mathbb{E}[\| \boldsymbol{M} - \sqrt{-\mathbf{H}} \hat{\boldsymbol{\delta}}\|^2]\) is a increasing function of \(\hat{\boldsymbol{\delta}}^\top \mathcal{I}(\hat{\boldsymbol{\delta}}) \hat{\boldsymbol{\delta}}\), the following measure is proposed:
\begin{align}
\label{eq:FisherMeasure}
\omega_{\nu j}\overset{\Delta}{=}&[\mathbf{e}^{(\nu)}_j\odot\boldsymbol{\beta}_\nu]^{\top}\boldsymbol{\mathcal{I}}_{\nu \nu} (\boldsymbol{\beta}_\nu)[\mathbf{e}^{(\nu)}_j\odot\boldsymbol{\beta}_\nu]\\
=&{\beta}^2_{\nu j} \cdot  \mathbb{E}\bigg[ \frac{\partial^2  \ell(\boldsymbol \delta) }{\partial \beta_{\nu j} \partial \beta_{\nu j} } \bigg]\notag,\,\,\nu=1,2,
\end{align}
where $\odot$ is the elements wise product and $[\mathbf{e}^{(\nu)}_j\odot\boldsymbol{\beta}_\nu]^{\top}\mathbf{1}_{W_\nu}=  \beta_{\nu j }$ is the parametric effect extracted from the coefficients vector  (containing smooth and parametric effects) associated with the $\nu^\text{th}$ margin, $ \mathbb{E}\bigg[ \frac{\partial^2  \ell(\boldsymbol \delta)}{\partial \beta_{\nu j} \partial \beta_{\nu j} } \bigg]$ is the element of the Fisher information matrix associated to the $j^{\text{th}}$ parametric effect.  

\begin{remark}
 Note that the Fisher information measure ($\omega_{\nu j}$) is always non-negative, as it is based on the model framework discussed in Section \ref{Ch:modelformulation}, which in turn is derived from an optimization procedure based on the trust-region algorithm. Consequently, the Fisher information measure, which depends on the diagonal elements of $\mathcal{I}(\boldsymbol{\delta})$, cannot take negative values. $\qed$ 
\end{remark}

The new measure presented in \eqref{eq:FisherMeasure} can isolate the contribution of the $j^{\text{th}}$ covariate to the $\nu^{\text{th}}$ survival function by accounting for the sharpness of the likelihood in the direction of the parametric effect. A high value of $ \mathbb{E}\bigg[ \frac{\partial^2  \ell(\boldsymbol \delta)}{\partial \beta_{\nu j} \partial \beta_{\nu j} } \bigg]$ denotes a sharp curvature of the likelihood, therefore greater accuracy in the estimate and vice versa \citep{amari2012differential}. In other words, we are not only attributing a high score to parameters different from zero but also taking into account the information in the bootstrap sample. An extensive discussion about the theoretical properties of \eqref{eq:FisherMeasure} and its consistency are presented in  Supplementary Material \ref{ch:teoretProp}. Once the measure is defined, a further challenge is the estimation of $\omega_{\nu j}$.

\begin{remark} [Estimation of $\omega_{\nu j}$ under the BRBVS framework]
\label{rk:estimationsetup}
    The computation of $\omega_{\nu j}$ for the class of Copula Link Based Survival Additive model(s) has two main complexities: i) avoiding the introduction of spurious correlation in the selection algorithm; ii) maintaining the computational burden as low as possible.

To address these challenges, under the BRBVS algorithm the following setup is proposed:
\begin{align*}
\eta_{1i}&=\beta_{1 0}+\sum_{k_{1}=1}^{K_{1}} s_{1 k_{1}}\left({t}_{1 k_{1} i}\right)+ \beta_{1j} {x}_{ij}\\
\eta_{2i}&=\beta_{2 0}+\sum_{k_{2}=1}^{K_{2}} s_{1 k_{2}}\left({t}_{2 k_{2} i}\right)+ \beta_{2j}{x}_{ij}\\
\eta_{3i}&=\beta_{3 0}\quad  i=1, \ldots, n; j=1, \dots, p,
\end{align*}
where we notice the baseline function of times and that $\boldsymbol{x}_j$ enters linearly in both $\eta_1$ and $\eta_2$.

This setup does not introduce spurious correlation in the selection algorithm since a) the same covariate is simultaneously evaluated in $\eta_1$ and $\eta_2$; b) we are not forcing a relationship between $\eta_1$ and $\eta_2$ through the introduction of covariates (that might have been already included in the two margins) in the specification of $\eta_3$. Furthermore, the computational complexity associated with this setup is $O(p)$. $\qed$
\end{remark}

\begin{remark} 
    Despite the estimation set-up discussed in Remark \ref{rk:estimationsetup}, the model framework is additive due to the monotonic splines used to model the baseline survivals. $\qed$
\end{remark} 

\begin{remark} [Effective degrees of freedom]
\label{rk:edfinCopulaModel}
    As pointed out by \cite{Marra2020}, the effective degrees of freedom (\(edf\)) for a model described by equation \eqref{eq:jointSurvival}, which contains only unpenalized terms, are equal to \(\xi\), the dimension of \(\boldsymbol{\delta}\). This   because  \(\mathrm{tr}(\sqrt{-\mathbf{H}}(-\mathbf{H} + \mathbf{S})^{-1}\sqrt{-\mathbf{H}}) = \mathrm{tr}(\mathbf{I})\), where \(\mathbf{H}\) is the Hessian matrix, \(\mathbf{S}\) is defined in equation \eqref{eq:penLik}, and \(\mathbf{I}\) represents the identity matrix. For a penalized model, the \(edf\) are given by \(\xi - \mathrm{tr}((- \mathbf{H} + \mathbf{S})^{-1}\mathbf{S})\). This latter quantity helps clarify the role of the smoothing parameter vector \(\boldsymbol{\lambda}\) (contained within \(\mathbf{S}\)). If \(\boldsymbol{\lambda} \rightarrow \mathbf{0}\), then \(\mathrm{tr}(\mathbf{A}) \rightarrow \xi\); conversely, if \(\boldsymbol{\lambda} \rightarrow \mathbf{\infty}\), then \(\mathrm{tr}(\mathbf{A}) \rightarrow \xi - \zeta\), where \(\zeta\) denotes the total number of the model's parameters subject to penalization. When \(\boldsymbol{\lambda}\) lies in the interval \((\mathbf{0}, \mathbf{\infty})\), where \(\mathbf{0}\) is the null vector, the model's \(edf\) equals a value within the range \([\xi - \zeta, \xi]\). $\qed$

\end{remark}

\begin{remark} [Role of effective degrees of freedom in the BRBVS algorithm]
\label{rk:edfinBRBVS}
Based on Remarks 1 and 2, the effective degrees of freedom $(edf)$ can be ignored as the estimated parameters remain constant across all bootstrap replicates. Furthermore, in a comprehensive examination of the information matrix, it is observed that the off-diagonal elements capture the cross-correlation among the parameters. Nevertheless, in our estimation framework, these components do not yield significant insights. Specifically, when scrutinizing each bootstrap sample, one assesses the contribution of the $j^{\text{th}}$ covariate across both survival margins. Consequently, any evaluation leveraging the off-diagonal components risks introducing spurious correlation, which remains inconsequential to the overarching objective of variable ranking. $\qed$
\end{remark}

\section{Simulation study}\label{Ch:Main_simulation}
To give evidence of the performance of the variable selection procedure and of the advantages that can arise from the use of the new measure introduced in Section \ref{eq:FisherMeasure}, we implement a simulation study where the proposed BRBVS algorithm is evaluated under different scenarios.

Let $\mathbf{X}\in \mathbb{R}^{n \times p}$ be the design matrix partitioned as $\mathbf{X}=(\mathbf{X}_{11} : \mathbf{X}_{12})$, where $\mathbf{X}_{11}$ contains the \emph{informative covariates}, 
$\mathbf{X}_{11}=(\boldsymbol{x}_1, \boldsymbol{x}_2, \boldsymbol{x}_3) \in \mathbb{R}^{n \times 3}$, with $\boldsymbol{x}_1, \boldsymbol{x}_2, \boldsymbol{x}_3$ three $n$-dimensional vectors, while the remaining $p-3$ variables included in $\mathbf{X}_{12}$ are  \emph{non-informative}.
In more detail, $\mathbf{X}_{11}$ and $\mathbf{X}_{12}$  are generated from a multivariate Gaussian distribution such that 
$\mathbf{X}_{11} \sim \mathcal{N}_3(\mathbf{0}, \boldsymbol{\Sigma}_{\mathbf{X}{11}})$, with $\mathbf{0}$ a null vector of means, and $\boldsymbol{\Sigma}_{\mathbf{X}_{11}}$ is the covariance matrix with diagonal elements fixed to one and off-diagonal elements equal to $0.5$;  $\mathbf{X}_{12} \sim \mathcal{N}_{p-3}(\mathbf{0}, \boldsymbol{\Sigma}_{\mathbf{X}_{12}})$  with $\boldsymbol{\Sigma}_{\mathbf{X}_{12}}$ an identity matrix.

 Let $n = \{800,1000\}$ (the reasons behind this choice are clarified in the following) and $p = \{100,200\}$ be the two sample sizes and the number of covariates in $\mathbf{X}$, respectively. 
Let $T_{\nu i}$ be the two time-to-events, with  $\nu=1,2$, where $T_{1i}$ is generated from a Proportional Hazards model, and $T_{2i}$ is derived from a Proportional Odds model, as shown in Table \ref{tab:SimMargins}.

\begin{sidewaystable}
\centering
\begin{tabular}{c|ll}

\hline
\textbf{Parameter} & \textbf{Model for $T_{1i}$} & \textbf{Model for $T_{2i}$} \\
\hline
Model Type & Proportional Hazards (\texttt{PH}) & Proportional Odds (\texttt{PO}) \\
\hline

Formula & $T_{1i}=\log[-\log{S_{10}(t_{1i})}]+\beta_{11}x_{1i}+\beta_{12}x_{2i}$ & $T_{2i}=\log\bigg[\frac{(1-S_{20}(t_{2i}))}{S_{20}(t_{2i})}\bigg]+\beta_{21}x_{1i}+\beta_{22}x_{3i}$ \\
$S(t)$ & $S_{10}(t_{1i})=0.9e^{(-0.4t_{1i}^{2.5})}+ 0.1 e^{(-0.1 t_{1i})}$ & $S_{20}(t_{2i})=0.9e^{(-0.4t_{1i}^{2.5})} + 0.1 e^{(-0.1 t_{1i})}$ \\
$\beta_{11}$ &$ -1.5 $& - \\
$\beta_{12}$ & $1.7 $& - \\
$\beta_{21}$ & - & $-1.5$ \\
$\beta_{22}$ & - & $-1.3$ \\
\hline
\end{tabular}
\caption{Margins generation process with $S(\cdot)$ the survival function, $\log(\cdot)$ the natural logarithm, and $e(\cdot)$ the corresponding base.  }
\label{tab:SimMargins}

\end{sidewaystable}

The times are generated with Brent's univariate root-finding method. This is accomplished by the \texttt{stats::unitroot} function in \texttt{R}, specifying a range between $0$ and $8$, and allowing for potential extensions.
The first step is to generate the times $T_1$. Subsequently, after observing the first time, the observations for $T_2$ are simulated. This implies that the covariate $\boldsymbol{x}_2$, although not appearing in the formula on the right in Table \ref{tab:SimMargins}, is certainly a relevant feature for $T_2$. From a causal inference perspective, the set of covariates involved in the data-generating process of $T_2$ is $\{\boldsymbol{x}_1, \boldsymbol{x}_2, \boldsymbol{x}_3\}$. Subsequently, the two survival times are combined through the Clayton copula. We consider two different scenarios of dependence between the two survival margins:
\begin{itemize}
    \item \textbf{Scenario A}: $\eta_{3i}=\beta_{30}$;
    \item \textbf{Scenario B}: $\eta_{3i}=\beta_{31}x_{1i}+\beta_{32}x_{2i}+\beta_{33}x_{3i}$,
\end{itemize}
with $\boldsymbol{\beta}_3=(\beta_{30}, \beta_{31}, \beta_{32}, \beta_{33})=(1.2, -1.5, 1.7, -1.5)$. The values of $\eta_{3i}$ are specified to range the Kendall tau value between $0.10$ and $0.90$. It is important to note that $\eta_{3i}$ represents the potential variability in the dependence between $(T_{1i}, T_{2i})$ across different observations.
Further, in Scenario A, we are considering weak dependence between $(T_{1i}, T_{2i})$, while in Scenario B, the data are generated by considering a dependence between $(T_{1i}, T_{2i})$.

Moreover, the data points experience right-censoring, resulting in $11\%$ and $32\%$ missing information for the first and second times, respectively. Specifically, random censoring times are derived from the lower and upper bounds of two independent uniform random variables. These bounds are then compared with the simulated times to assign censoring.

Given the computational complexity of the BRBVS algorithm and the time needed to obtain the estimates for model  \eqref{eq:jointSurvival}  from the \texttt{gjrm()} function in the
 \texttt{GJRM} package in \texttt{R} \citep{Man:GJRM}, the simulation study is based on $n_{\text{rep}}=100$ Monte Carlo replicates and for each run we consider $B=50$ bootstrap replicates. 
The threshold parameter $\tau$ in \ref{eq:hat_s_v} is set at $\tau=0.5$, whereas the maximum number of important covariates is $k_{\text{max}}=6$. This set of parameters $(B, \tau, k_{\text{max}})$ is determined following \cite{Baranowski2020} and more details are available in Supplementary Material \ref{Ch:Simulation}.

In the model class discussed in Section \ref{Ch:modelformulation}, the sample size significantly influences the precision of the estimates, as highlighted in Supplementary Material G in \cite{PETTI22}. To achieve this precision, we select the two sample sizes  $n=\{800, 1000\}$, to allow bootstrap samples with an appropriate number of units. For this aim, we fix $r=2$ (as in \cite{Baranowski2020}) and consequently  the bootstrap subsets contain $m=\{400, 500\}$ units, respectively.

The measures used to rank the variables are two:\\
(a) $\hat{\omega}_{\nu j}$ based on the Fisher information matrix, proposed in Section \ref{SubCh:FisherMetric};\\
(b) the absolute value of the estimated coefficients $|\hat{\beta}_{\nu_j}|$, denoted by $\hat\phi_{\nu j}$, for $(\nu=1,2; j=1,\dots, p)$. 

This last measure is largely used in the regression context to screen variables even in the presence of datasets of large dimension (see among the others \citep{FanSong2010}) but, differently from the proposed measure (a), it considers only the punctual value of the estimate both discarding other information that can be taken from the likelihood and neglecting the intrinsic dependence of the copula function. 

As said before, the coefficients are estimated using the \texttt{gjrm()} function from the \texttt{GJRM} package   \citep{Man:GJRM} and then the two measures (and consequently the variable ranking) are based on its output.

In more detail, for each Monte Carlo replicate we perform the following steps:
\begin{enumerate} 
\setlength\itemsep{-0.5em}
\item Generate the design matrix $\boldsymbol X$, as clarified before;
\item Generate the margins following the scheme in Table \ref{tab:SimMargins};
\item Generate the matrix $\boldsymbol Z$, with rows $\mathbf z_i=(t_{i1}, t_{i2}, x_{i1}, \ldots, x_{ip})$;
\item For each bootstrap replicate:
\begin{enumerate} 
\setlength\itemsep{-0.5em}
\item generate $r$ subsets from $\boldsymbol Z$;
\item for each subset estimate the coefficients of model \eqref{eq:jointSurvival}  using the \texttt{GJRM:gjrm()} function and taking the information matrix;
\item compute the two measures,  $\hat{\omega}_{\nu j}$ and $\hat\phi_{\nu j}$, for $j=1,2,\ldots p$ and $\nu=1,2$, and rank the covariates according to the associated measures;
\end{enumerate}
\item  use the ranking obtained from the $B$ bootstrap replicates to estimate the probabilities \eqref{eq:hat_pi_v} and define the top ranked covariates for each $\eta_\nu$, $\nu=1,2$;
\item select the set of important variables $\hat{s}_\nu$, for $\nu=1,2$ using \eqref{eq:hat_s_v}.
\end{enumerate}

Finally, note that from steps 4.3 to 6, the ranking, computation of the probabilities, and selection of the important variables are replicated for both measures, whose performance is evaluated and compared. 
 
The performance of the BRBVS  obtained through the \texttt{BRBVS()} function in \texttt{R} \citep{Petti2024Package} is assessed by computing different metrics based on the evaluation of False Positive (FP) and False Negative (FN), where FP is the number of covariates incorrectly chosen as relevant by the variable selection procedure, and  FN is the number of covariates incorrectly chosen as irrelevant by the variable selection procedure. Then, we compute the average of FP and of FN across all $n_{\text{rep}}$ replicates. Specifically, given the average of FP  is equal to $FP_{\nu}=n_{\text{rep}}^{-1}\sum_{h=1}^{n_{\text{rep}}} FP_{\nu}^{(h)}$ and the average of FN is $FN_{\nu}=n_{\text{rep}}^{-1}\sum_{h=1}^{n_{\text{rep}}} FN_{\nu}^{(h)}$, $FP_{\nu}^{(h)}$ and $FN_{\nu}^{(h)}$ denote the number of false positive and false negative, respectively, for the $\nu^{\text{th}}$ survival outcome in the $h^{\text{th}}$ replicate, for $\nu=1,2$.

Additionally, the average number of variables selected in the estimated relevant set for the $\nu^{\text{th}}$ survival is represented by $\langle\hat{{s}}_\nu\rangle=n_{\text{rep}}^{-1} \sum_{h=1}^{n_{\text{rep}}} |\hat{s}^{(h)}_{\nu}|$, where $|\hat{s}^{(h)}_{\nu}|$ is the estimated size of relevant variables in the $h^{\text{th}}$ replicate.

Furthermore, denote with $s_\nu$, for $\nu=1,2$, the true sets of relevant variables for the first and second survival outcomes, where in our study $s_1=\{1,2\}$ and $s_{2}=\{1, 2, 3\}$. The average number of correctly identified relevant variables across replicates is given by $\langle\hat{s}_{\nu}\cap s_\nu\rangle=n_{\text{rep}}^{-1} \sum_{h=1}^{n_{\text{rep}}}|\hat{s}^{(h)}_{\nu}\cap s_\nu|$, where $|\hat{s}^{(h)}_{\nu}\cap s_\nu|$ counts the relevant variables correctly selected in each replicate.

In {\em Scenario A}, where $\eta_{3i}$ is only a constant value $\beta_{30}$, and the dependence between $T_{1i}$ and $T_{2i}$ is mainly due to $x_{1i}$ and $x_{2i}$, $i=1,2,\ldots, n$, in Table \ref{tab:performance} it can be noted that the two measures mainly differ for the number of the  selected variables in the first margin: the average number of selected variables 
($\langle\hat{s}_1\rangle$) is always lower for the $\hat\omega$ metric with respect to $\hat\phi$ and consequently the $FP_1$ is also lower for all $n$ and $p$.\\
More interesting results in Table \ref{tab:performance} are for {\em Scenario B}, where $\eta_{3i}$ includes all variables used to model $T_{1i}$ and $T_{2i}$. Despite the higher complexity of the model, the mean value of correctly selected relevant covariates, $\langle\hat{s}_{\nu}\cap s_\nu\rangle$ for $\nu=1,2$, is quite satisfactory with both metrics and
the $\hat\omega$ measure performs always better than the $\hat\phi$ metric in terms of $FP_1$. Furthermore, the $\hat\phi$ measure shows an unsatisfactory randomness when $n=800$ and $p$ goes from 100 to 200, giving evidence of its incorrectness with the model under analysis.

The results of Table \ref{tab:performance} are also illustrated in Figure \ref{fig:PlotScenarioA} and Figure \ref{fig:PlotScenarioB} in the Supplementary Material where the mean values of the FP and FN are shown for both scenarios and for different values of $n$ and $p$.

More results for both scenarios are presented and discussed in the Supplementary Material \ref{Ch:Simulation}. They give further evidence of the advantages that may arise from the BRBVS algorithm combined with the proposed measure.

\begin{table}[h]
\centering
\begin{tabular}{l>{\columncolor{darkgray}}c>{\columncolor{lightgray}}c>{\columncolor{darkgray}}c>{\columncolor{lightgray}}c>{\columncolor{darkgray}}c>{\columncolor{lightgray}} c>{\columncolor{darkgray}} c>{\columncolor{lightgray}}c}
  \hline
  \multicolumn{9}{c}{$\text{\textbf{Scenario A}}: \eta_{3i}=\beta_{30}$ } \\  
  \hline 
  &\multicolumn{4}{c}{$n=800$} & \multicolumn{4}{c}{$n=1000$}  \\ 
  & \multicolumn{2}{c}{$p=100$}&  \multicolumn{2}{c}{$p=200$}&  \multicolumn{2}{c}{$p=100$}  &  \multicolumn{2}{c}{$p=200$} \\ 
  \hline
& $\hat\omega$ & $\hat\phi$ &$\hat\omega$ & $\hat\phi$ & $\hat\omega$ & $\hat\phi$ & $\hat\omega$ &$\hat\phi$\\ 
 \hline
$\langle\hat{s}_1\rangle$ & 2.1800 & 2.2400 & 2.4000 & 2.4800 & 2.0800 & 2.3400 & 2.1800 & 2.3200 \\  
  $\langle\hat{s}_2\rangle$ & 3.0000 & 2.9800 & 2.9600 & 2.9200 & 3.0000 & 3.0000 & 3.0000 & 2.9400 \\
   $\langle\hat{s}_{1}\cap s_1 \rangle$  & 1.9800 & 2.0000 & 2.0000 & 2.0000 & 2.0000 & 2.0000 & 2.0000 & 2.0000 \\ 
    $\langle\hat{s}_{2}\cap s_2 \rangle$  & 3.0000 & 2.9800 & 2.9600 & 2.9200 & 3.0000 & 3.0000 & 2.9600 & 2.9400 \\ 
$FP_{1}$ & 0.0333 & 0.0400 & 0.0667 & 0.0800 & 0.0133 & 0.0567 & 0.0300 & 0.0533 \\ 
  $FP_{2}$& 0.0000 & 0.0000 & 0.0000 & 0.0000 & 0.0000 & 0.0000 & 0.0067 & 0.0000 \\ 
  $FN_{1}$& 0.0050 & 0.0000 & 0.0000 & 0.0000 & 0.0000 & 0.0000 & 0.0000 & 0.0000 \\ 
 $FN_{2}$&0.0000 & 0.0050 & 0.0100 & 0.0200 & 0.0000 & 0.0000 & 0.0100 & 0.0150 \\

\hline

\end{tabular}
\caption{Summary of the simulation results for the two scenarios, A and B. $FP_{\nu}$: average of False Positives;  $FN_{\nu}$: average number of False Negatives (for $\nu=1,2$ and across all simulation replicates, $n_{\text{rep}}$). 
$\langle\hat{s_\nu}\rangle$: average number of variables selected for each survival outcome.
$\langle\hat{s}_{\nu}\cap s_\nu\rangle$ average count of correctly identified relevant variables. 
 The columns  $\hat\omega$ denote the cases where the BRBVS measure is based on the Fisher Information described in Section \ref{SubCh:FisherMetric}
 whereas the columns $\hat\phi$ are the cases where the BRBVS measure is the absolute value of the estimated coefficients, $\phi_{\nu j}=|\beta_{\nu_j}|$  $(\nu=1,2; j=1,\dots, p)$.}
 \label{tab:performance}
\end{table}

\begin{table}[h]
\ContinuedFloat 
\centering
\begin{tabular}{l>{\columncolor{darkgray}}c>{\columncolor{lightgray}}c>{\columncolor{darkgray}}c>{\columncolor{lightgray}}c>{\columncolor{darkgray}}c>{\columncolor{lightgray}} c>{\columncolor{darkgray}} c>{\columncolor{lightgray}}c}
  \hline
 \multicolumn{9}{c}{$\text{\textbf{Scenario B}}: \eta_{3i}=\beta_{31}x_{1i}+\beta_{32}x_{2i}+\beta_{33}x_{3i} $} \\  
  \hline 
  &\multicolumn{4}{c}{$n=800$} & \multicolumn{4}{c}{$n=1000$}  \\ 
  & \multicolumn{2}{c}{$p=100$}&  \multicolumn{2}{c}{$p=200$}&  \multicolumn{2}{c}{$p=100$}  &  \multicolumn{2}{c}{$p=200$} \\ 
  \hline
  & $\hat\omega$ & $\hat\phi$ &$\hat\omega$ & $\hat\phi$ & $\hat\omega$ & $\hat\phi$ & $\hat\omega$ &$\hat\phi$\\ 
 \hline

 $\langle\hat{s}_1\rangle$  &2.0600 & 2.2600 & 2.0600 & 2.1600 & 2.0200 & 2.2200 & 2.1200 & 2.4000 \\ 
  $\langle\hat{s}_2\rangle$&  2.9800 & 3.0200 & 3.0600 & 3.0000 & 3.0200 & 2.9800 & 2.9400 & 2.9600 \\
   $\langle\hat{s}_{1}\cap s_1 \rangle$ & 2.0000 & 2.0000 & 2.0000 & 2.0000 & 2.0000 & 2.0000 & 2.0000 & 2.0000 \\ 
   $\langle\hat{s}_{2}\cap s_2 \rangle$ & 3.0000 & 2.9800 & 2.9600 & 2.9800 & 2.9200 & 2.9400 & 3.0000 & 2.9800 \\ 
 
  $FP_{1}$ & 0.0100 & 0.0433 & 0.0100 & 0.0267 & 0.0033 & 0.0367 & 0.0200 & 0.0667 \\ 
  $FP_{2}$ & 0.0033 & 0.0067 & 0.0100 & 0.0033 & 0.0033 & 0.0000 & 0.0033 & 0.0033 \\ 
  $FN_{1}$& 0.0000 & 0.0000 & 0.0000 & 0.0000 & 0.0000 & 0.0000 & 0.0000 & 0.0000 \\  
 $FN_{2}$& 0.0100 & 0.0050 & 0.0000 & 0.0050 & 0.0000 & 0.0050 & 0.0200 & 0.0150 \\ 

   \hline
\end{tabular}
\centering
\caption*{Table \thetable: Cont.} 

\end{table}


\clearpage
\section{Empirical results of variable selection on AREDS dataset } \label{Ch:Application}

In this section, we evaluate the BRBVS algorithm performance by analyzing a subsample of 314 patients from the AREDS dataset, which is openly accessible through the \texttt{BRBVS} package. This dataset has been recently analyzed in \cite{PETTI22, sun2021copula}. Some details on the dataset and the \texttt{R} code used in this section can be found in the Supplementary Material \ref{ch:Rsoftware}.

The AREDS dataset 
has been built to investigate on the clinical progression of macular degeneration (AMD) and cataract of a sample of patients. It
includes the start and end times of observation for both eyes, denoted by \( \texttt{enroll\_day} \) (start of follow-up in days) and \( \texttt{AMD\_recurrence} \) (time to recurrence or last follow-up in days), along with the associated censoring variables. The covariates in the analyzed subsample are the Severity Score (a variable ranging from 4 to 8, where a higher value indicates more severe Age-related Macular Degeneration (AMD)), with \texttt{SevScale1E} representing the score for the right eye and \texttt{SevScale2E} for the left eye; \texttt{ENROLLAGE}, denotes the age at enrollment (a numerical variable); and the genetic variant \texttt{rs2284665}, a categorical variable with levels defined as $0$ (\texttt{GG}), $1$ (\texttt{GT}), and $2$ (\texttt{TT}).

To assess the effectiveness of the Bivariate Ranking-Based Variable Selection (BRBVS) algorithm to select two sets of relevant covariates, the AREDS dataset is perturbed by including \(100\) independent instances of a standard Gaussian distribution.
The parameters of the algorithm are then specified as \(k_{\text{max}}=10\), \(m=328\), \(\tau=0.5\), bootstrap samples \(B=100\), Plackett copula (\texttt{PL}) and Proportional odds (\texttt{PO, PO}). Details on the choice of the copula margins configuration for the AREDS data are in Supplementary Material E in \cite{PETTI22}.

\begin{remark}
The copula function $C(\cdot)$ and the link functions $g_\nu(\cdot)$ can be selected by AIC and BIC in a standard fashion, as with any other canonical model (e.g. linear regression model, generalized linear model). In the BRBVS package, we implemented the function \texttt{Select\_link\_BivCop()} to select the best link functions. $\qed$
\end{remark}

The BRBVS is then applied under two different measures:
\begin{itemize}
\item[(a)] the measure based on the Fisher Information matrix \( \hat{\omega} \), presented in Section \ref{SubCh:FisherMetric};
\item[(b)] the absolute value of the estimated coefficients $|\hat{\beta}_{\nu j}|$, denoted with $\hat{\phi}_{\nu j}$ for $j=1,\ldots, p$ and $\nu=1,2$.
\end{itemize}
With the Fisher information measure \( \hat{\omega} \), the selected sets of variables are $ \hat{s}_1=\{\texttt{SevScale1E} $ $ \texttt{SevScale2E}, \texttt{ENROLLAGE}\} $ and \( \hat{s}_2=\{\texttt{SevScale1E}, \texttt{SevScale2E}\} \) as shown in Table \ref{tab:BRBVSrho0}, while with the absolute values of the coefficients \( \hat{\phi} \), the same sets of relevant covariates are chosen for \( \eta_1 \) and \( \eta_2 \), namely \(\{ \texttt{SevScale1E}, \texttt{SevScale2E} \}\).

The main difference is that the measure proposed in Section \ref{SubCh:FisherMetric} selects \( \texttt{ENROLLAGE} \), which is ignored when the absolute value of the coefficients is adopted. Furthermore,  \( \hat{\phi} \) selects consistently two noise variables $\mathbf{x}_{41}$ and  $\mathbf{x}_{61}$ in the $57\%$ and $78\%$ of the bootstrap replicates, respectively.

\begin{table}[t]
\footnotesize
   \centering
    \begin{tabular}{lcc|ccc}
     \multicolumn{5}{c}{Copula \texttt{PL}, margins \texttt{PO,PO}} \\
  
    \hline
     \backslashbox{Features}{Measure} & \makebox[3em]{\( \hat{\omega}_1 \)} & \makebox[3em]{\( \hat{\omega}_2 \)} & \makebox[3em]{\( \hat{\phi}_1 \)} & \makebox[3em]{\( \hat{\phi}_2 \)} \\
    \hline
      \texttt{SevScale1E} & $51\%(1^{\text{st}})$ & $46\% (1^{\text{st}})$ & $94\% (2^{\text{nd}})$ & $72\% (2^{\text{nd}})$ \\
           \texttt{SevScale2E} & $91\% (2^{\text{nd}})$  &$91\% (2^{\text{nd}})$   & $50\% (1^{\text{st}})$   & $39\% (1^{\text{st}})$  \\
       \texttt{ENROLLAGE}  &$22\% (3^{\text{rd}})$   & -  &-   & -  \\
       \texttt{GG}         &   &   &   &   \\
       \texttt{GT}         &   &   &   &   \\
       \texttt{TT}         &   &   &   &   \\
      \texttt{$\mathbf{x}_{41}$} & -  & -  & -  & $57\% (3^{\text{rd}})$  \\
       \texttt{$\mathbf{x}_{61}$} &  - & -  & -  &  $78\% (4^{\text{th}})$ \\
    \hline
    \end{tabular}
    \caption{BRBVS results using copula (\texttt{PL}), Proportional Odds (\texttt{P0}), and Proportional Odds (\texttt{P0})  for the first and second survival respectively, with parameters set as follows: \( k_{\text{max}}=10 \), \( m=328 \), \( \tau=0.5 \), \( B=100 \). The Fisher information measure \eqref{eq:FisherMeasure} is denoted as \( \hat{\omega}_{ \nu}\), the measure based on the absolute value of the estimates as \( \hat{\phi}_{ \nu}\), $\nu=1,2$. The variable ranking is reported in parentheses, whereas the percentages refer to the empirical frequency of the variable selection. The covariates $\mathbf{x}_{41}$ and $\mathbf{x}_{61}$ are the nuisance variables, as discussed in this section.}
\label{tab:BRBVSrho0}
\end{table}

Once the relevant covariate sets have been identified, we can estimate the model by the \texttt{gjrm()} function. In detail, for the marginal equation $\eta_1$, the smooth functions of \texttt{ENROLLAGE} and the time variables are represented using penalized thin plane regression splines with second order penalty \citep{Wood2017} and the monotonic penalized B-splines, respectively.   Regarding \( \eta_2 \), the time variable is modeled similarly to \( \eta_1 \). The remaining variables, \texttt{SevScale1E} and \texttt{SevScale2E}, enters linearly in \( \eta_1 \), \( \eta_2 \), and \( \eta_3 \). The severity score can not accommodate a spline due to its discrete nature. Given this setting and the technical specifications of the PC (a 2.3-GHz Intel(R) Core(TM) computer running macOS version 13.6.1), the average computing time to fit a model is about 16 seconds, and the length of the model parameter vector is 54.

The estimation results in Table \ref{tab:estimation} show that the estimated values associated with the severity scores grow in conjunction with the increase in severity score. More precisely, for the first margin, the parameters tied to the severity scores of the first eye, listed in ascending order of severity, are 0.69, 0.80, 1.72, and 2.59. In contrast, the corresponding parameters for the second eye are 0.40, 0.87, 1.02, and 1.48. Regarding the second margin, the severity parameters for the first eye are 0.27, 0.34, 0.85, and 0.99, whereas for the second eye, they are 0.87, 1.23, 2.25, and 3.43.

This result confirms that, as expected, a high severity score is associated with an increased risk in the disease progression and vice versa. Furthermore, the spline on age in the first margin is also significant, confirming the non-linearity of age in the progression of the disease. 

Figure \ref{fig:splinesAREDS} displays the splines for \texttt{ENROLLAGE} and for the first eye (middle plot), revealing that disease progression tends to accelerate after the age of 65. The sample predominantly consists of individuals aged between 61 and 77, with very few observations falling outside this range. The estimated smooth functions of time (external plots) exhibit increasing monotonic trends, consolidating the idea that the risk increases with time.

 \begin{table}
\begin{center}
{\footnotesize
\scalebox{0.9}{
\begin{tabular}{l|cccccc}
\multicolumn{1}{r}{}
& \multicolumn{1}{l}{First Margin}
& \multicolumn{1}{l}{}
& \multicolumn{1}{l}{Second Margin}
& \multicolumn{1}{l}{}
& \multicolumn{1}{l}{Copula Margin} 
& \multicolumn{1}{l}{}\\
\cline{2-7}
Parametric Eff. & Estimate (S.e.)  & P-val& Estimate (S.e.) & P-val & Estimate (S.e.)  & P-val\\
\hline
\texttt{(Intercept)} & -18.3487 (4.3956) & $<$0.001
& -30.4468(10.5514) & $<$0.01 
& 0.9658 (0.5089) & 0.0577  \\ 
 \texttt{SevScale1E5} & 0.6875 (0.2645) & $<$0.01 
  & 0.2684(0.2572) & 0.2967
  & -0.2302 (0.5196) & 0.6576 \\ 
  \texttt{SevScale1E6} & 0.8084 (0.2552) & $<$0.01
  & 0.3454(0.2459) & 0.1600
  & -0.6698 (0.4964) & 0.1772 \\ 
  \texttt{SevScale1E7} & 1.7219 (0.2750) & $<$0.001 
  & 0.8530 (0.2592) & $<$0.001
  & -1.1222 (0.5165) & $<$0.05 \\ 
  \texttt{SevScale1E8} &2.5884 (0.3608) & $<$0.001
  & 0.9947(0.3286)& $<$0.01
  & -1.0086 (0.6035) & 0.0946 \\ 
  \texttt{SevScale2E5} & 0.4063 (0.2813)  & 0.1448 
  & 0.8682 (0.2796) & $<$0.01
  & 1.2973 (0.5464) & $<$0.05 \\ 
  \texttt{SevScale2E6} &0.8768 (0.2692) & $<$0.01
  &  1.2294 (0.2747) & $<$0.001
  & 1.9442 (0.5167) & $<$0.001 \\ 
  \texttt{SevScale2E7} & 1.0164 (0.2900) & $<$0.001 
  & 2.2515 (0.2988) & $<$0.001
  &  1.7939 (0.5365) & $<$0.001  \\ 
  \texttt{SevScale2E8} & 1.4862 (0.3403) & $<$0.001
  & 3.4251 (0.3632) & $<$0.001
  & 1.1878 (0.6251) & 0.0574 \\ 
   \hline
   &&&&&&\\
   Smooth Eff. & edf  & P-val& edf & P-val &   & \\
   \hline
   \texttt{s(t)}&6.649 &$<$0.001&7.374&$<$0.001&&\\
   \texttt{s(ENROLLAGE)}&1.625&$<$0.01&-&-&&\\
\hline
\hline
\end{tabular}}
}
\end{center}

\caption{AREDS results. Estimates of the parameters with standard errors (S.e.) and the effective degrees of freedom (edf) obtained through the \texttt{gjrm()} function.}
\label{tab:estimation}
\end{table}
The top panels in Figure \ref{fig:jointSurvival} display the joint progression-free probability contours for subjects with severity scores of 4 (top-right), 6 (top-middle), and 8 (top-left) in both eyes, with ages represented by the set ${56,69,81}$. From a comparison of the first three plots on the top, it is evident that progression accelerates with age. This is particularly noticeable when examining the contours for individuals aged 56 and 69, where the difference in progression-free probability for identical severity scores becomes more distinct.

In the bottom panels in Figure \ref{fig:jointSurvival}, the contours are depicted for subjects aged 56 (bottom-right), 69 (bottom-middle), and 81 (bottom-left), with severity scores within the range of ${4,6,8}$. There is a clear distinction among the three severity score groups, each of which varies with age. Notably, the group with the highest AMD severity shows the lowest progression-free probabilities. 

The Kendall's tau coefficient, estimated with the sets of covariates identified by the BRBVS, is equal to 0.35 and indicates a weak correlation between the two margins. Remarkably, this value remains unchanged even when implementing a full model configuration, i.e., including \texttt{ENROLLAGE}, \texttt{SevScale1E}, and \texttt{SevScale2E}  in \( \eta_1 \), \( \eta_2 \), and \( \eta_3 \), with a total of 72 parameters.

Our findings are in agreement with the discussions presented in \cite{PETTI22} and \cite{sun2021copula}, as well as with \cite{shih2012development}, which demonstrate a significant correlation between age and the likelihood of disease progression. Additionally, they are also in line with \cite{armstrong2004illustration}, where the importance of the severity score in explaining the progression of AMD is discussed.

\begin{figure}[t]
    \centering
    \includegraphics[width=0.30\textwidth, angle=-90]{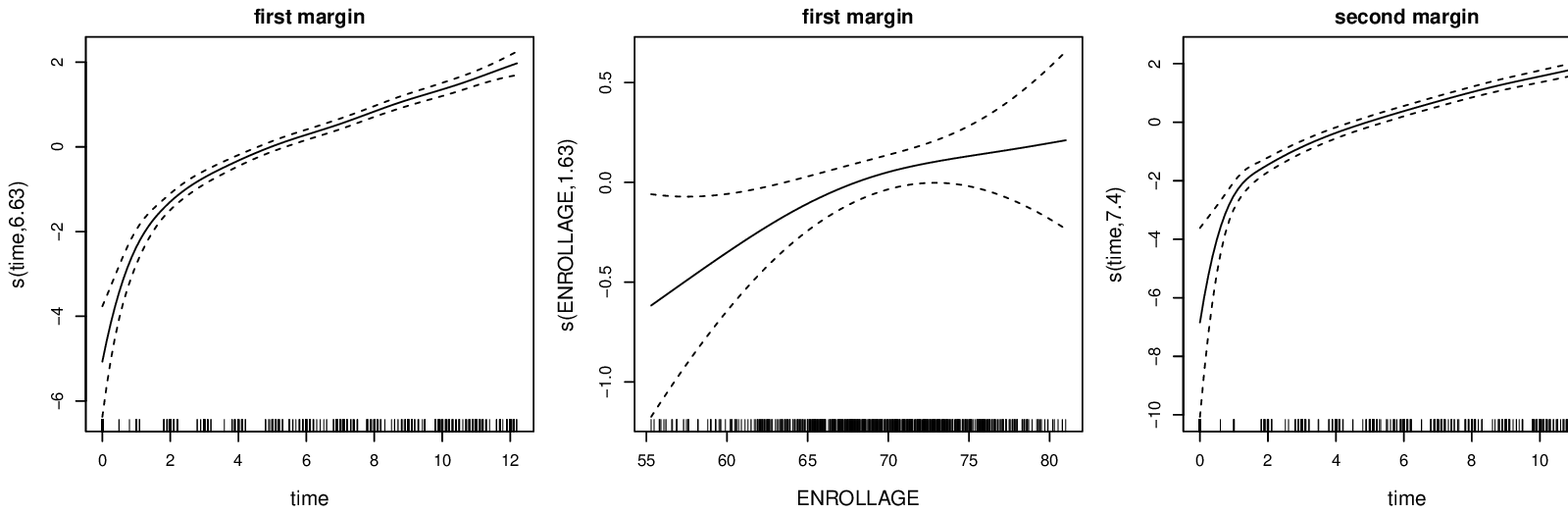}
    \caption{AREDS results. Baseline survival for the first and second eyes and the smooth function estimated for ENROLLAGE. The dotted line represents a 95\% point-wise interval. The number in brackets on the x-axis represents the estimated number of degrees of freedom (e.g., for the first survival, it is 6.63).}
    \label{fig:splinesAREDS}
\end{figure}

\begin{figure}[ht]
    \centering
    \includegraphics[width=0.63\textwidth, angle=-90]{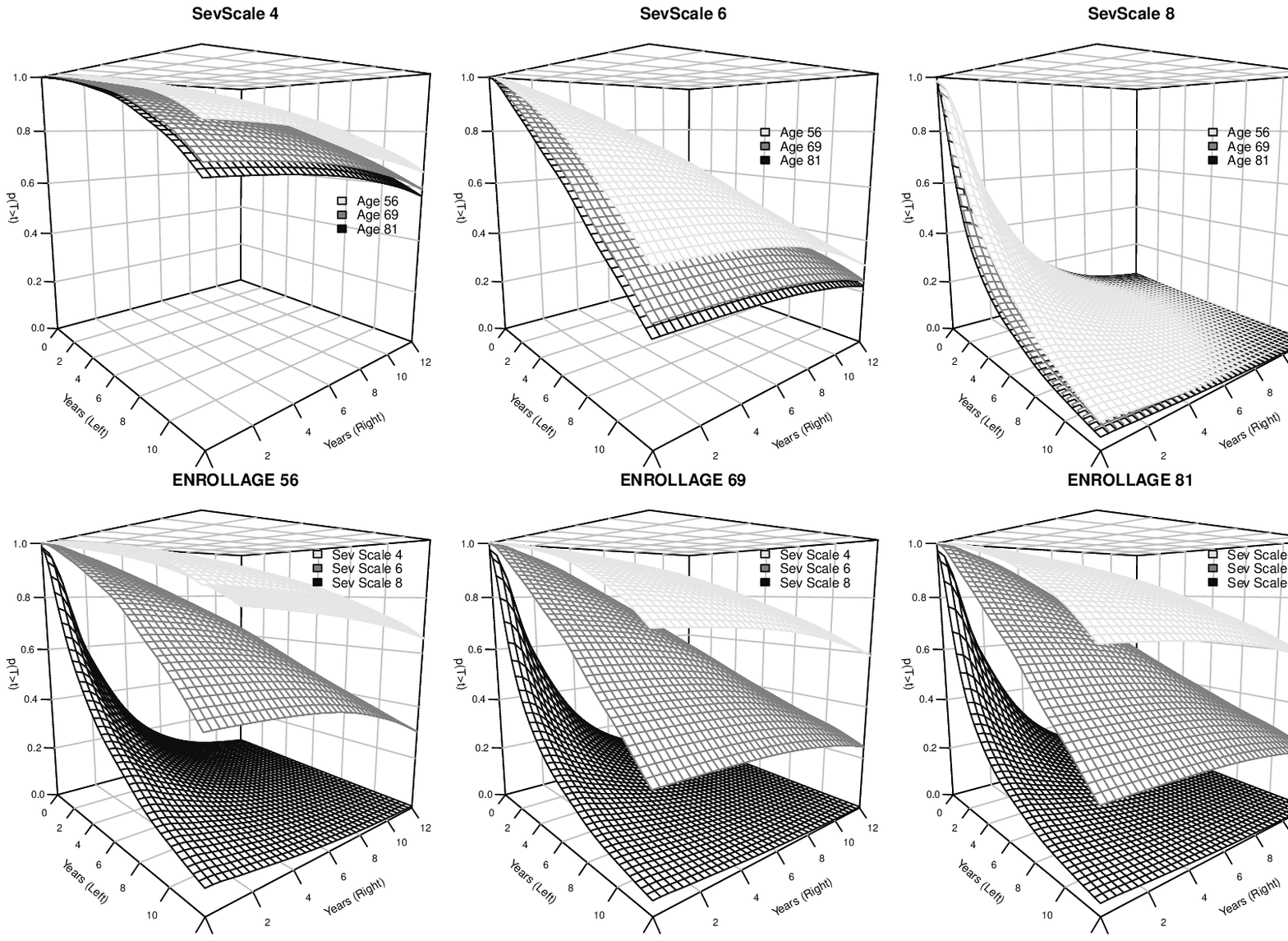}
    \caption{AREDS results. Joint progression-free probability contours for progression to late-AMD disease (in years) in the left and right eyes under different scenarios. In the top sets of panels, the severity score for both margins is set to 4 (top left),6 (top middle), and 8 (top right), varying with age. In the bottom panels, age is set to 56 (bottom left), 69 (bottom middle), and 81 (bottom right) with varying severity scores.}
    \label{fig:jointSurvival}
\end{figure}

\clearpage
\section{Discussion}
\label{ch:conclusion}

In this paper, we proposed a ranking algorithm established on a new measure for variable selection within the relevant covariates under Copula Link-Based Additive Models. The algorithm is based on the marginal evaluation of each covariate under this class of models, and its application can be extended even to the high-dimensional domain. Further, it is computationally sound, and ready to implement by practitioners in the field.

Our proposal introduces three main innovative elements: 1) to our knowledge, this is the first attempt to address the variable selection within the bivariate survival model \cite{Marra2020} domain; 2) we extended the RBVS algorithm of \cite{Baranowski2020}, originally proposed for univariate response variables, to the Copula link-based additive models for bivariate time-to-event analysis, leading to the development and implementation of the new extension of RBVS algorithm, called Bivariate RBVS (BRBVS) algorithm; 3) we proposed a novel measure for ranking variables within the BRBVS algorithm, highlighting both its theoretical underpinnings and empirical performance.

All developments are integrated into a new \texttt{R} package, called \texttt{BRBVS}, which allows both numerical and graphical summaries. 

We validated the BRBVS algorithm and evaluated the performance of the newly introduced ranking measure through a simulation study and the analysis of the AREDS dataset. Both analyses aim to assess the proposed measure and compare its performance with that of other measures for variable ranking largely applied to screen and/or select covariates in the regression framework.

The simulation results unequivocally showed the advantages afforded by our proposed algorithm and metric, underscoring its potential to enhance the variable selection process within the BRBVS algorithm.

Future research will be addressed to find more efficient variable selection algorithms by evaluating the relevant covariates via causal inference methods.  Furthermore, given the flexibility of the BRBVS algorithm and of the measure here proposed, it can be extended to a larger class of bivariate models that may potentially be included also in the package. Further opportunities for research may include developing a copula selection method, especially in the case of rotation, and extending the model to cases with more than two time-to-events. However, at this point, the current model framework has not yet been extended to handle more than two time-to-events.


\clearpage

\renewcommand{\appendix}{}
\renewcommand{\appendixname}{S}

\renewcommand{\thesection}{S\arabic{section}}
\setcounter{section}{0}

\renewcommand{\thepage}{S\arabic{page}}
\setcounter{page}{1}

\renewcommand{\thetable}{S\arabic{table}}
\setcounter{table}{0}

\renewcommand{\thefigure}{S\arabic{figure}}
\setcounter{figure}{0}

\renewcommand{\theequation}{S\arabic{equation}}
\setcounter{equation}{0}

\renewcommand{\thetheorem}{\arabic{theorem}}
\setcounter{theorem}{1}

\begin{appendix}

\begin{center}{\Large Supplementary Material\\\vspace{8pt}
Bivariate Variable Ranking for censored time-to-event data via Copula Link Based Additive models}
\end{center}

\section*{Supplementary material}

\section{Theoretical results }\label{ch:teoretProp}
This section will present the main theoretical results for the consistency of both the Fisher Information measure proposed and in the BRBVS procedure.

\begin{theorem}
\label{th:upperbound}
Let us consider $\omega_{\nu j}\overset{\Delta}{=}{\beta}^2_{\nu j} \cdot \mathbb{E}\left[ \frac{\partial^2 \ell(\boldsymbol \delta) }{\partial \beta_{\nu j} \partial \beta_{\nu j} } \right]$ for $\nu=1,2$, wherein $\beta_{\nu j}$ denotes the population parameter associated to the model framework discussed in Section \ref{Ch:modelformulation}, corresponding to the $\nu^{\text{th}}$ margin and $j^{\text{th}}$ covariate. Furthermore, $\mathbb{E}\left[ \frac{\partial^2 \ell(\boldsymbol \delta) }{\partial \beta_{\nu j} \partial \beta_{\nu j} } \right]\overset{\Delta}{=} i_{\nu \nu}(\beta_{\nu j})$ represents the associated element of the Fisher Information Matrix. Let $\hat{\omega}_{\nu j}$ denote an estimate of $\omega_{\nu j}$. Using the results in Theorem 1 (p. 15) and Remark 13 (p. 30) in  Supplementary Material B in \cite{Marra2020}, for some $\theta \in [0, \frac{1}{2})$, $\lambda_{\theta}>0$, and for any given $m>0$, the following inequality holds:

\[
\mathbb{P}\left( |\hat{\omega}_{\nu j} - \omega_{\nu j}| \geq \lambda_{\theta} m^{-\theta} \right) \leq  4 \exp(-C_{\theta} m^{1-2\theta}),
\]

where $\theta \in [0, \frac{1}{2})$, $C_{\theta}>0$.

\end{theorem}

\begin{proof}

Considering a first-order Taylor expansion of \(\hat{\omega}_{\nu j}\) around the true parameters \(\beta_{\nu j}\) and \(i_{\nu \nu}(\beta_{\nu j})\), we can write
\[
\hat{\omega}_{\nu j} - \omega_{\nu j} \approx \frac{\partial \omega_{\nu j}}{\partial \beta_{\nu j}} (\hat{\beta}_{\nu j} - \beta_{\nu j}) + \frac{\partial \omega_{\nu j}}{\partial i_{\nu \nu}} (\hat{i}_{\nu \nu}(\hat{\beta}_{\nu j}) - i_{\nu \nu} (\beta_{\nu j})),
\]

where $\omega_{\nu j} = \beta_{\nu j}^2 i_{\nu \nu} (\beta_{\nu j})$,$ \quad \frac{\partial \omega_{\nu j}}{\partial \beta_{\nu j}} = 2 \beta_{\nu j} i_{\nu \nu}(\beta_{\nu j})$,$ \quad\frac{\partial \omega_{\nu j}}{\partial i_{\nu \nu} (\beta_{\nu j})} = \beta_{\nu j}^2$. Therefore
\[
|\hat{\omega}_{\nu j} - \omega_{\nu j}| \leq 2 |\beta_{\nu j} i_{\nu \nu}(\beta_{\nu j})| |\hat{\beta}_{\nu j} - \beta_{\nu j}| + |\beta_{\nu j}^2| |\hat{i}_{\nu \nu} - i_{\nu \nu}| + O(||\hat{\omega}_{\nu j}- \omega_{\nu j}||^2).
\]

Given that \(\hat{\beta}_{\nu j}\) is asymptotically normal (Remark 13 in Supplementary Material B in \cite{Marra2020}), for any \(\varepsilon > 0\), we can write \citep{van2000asymptotic}

\[
\mathbb{P}\left( |\hat{\beta}_{\nu j} - \beta_{\nu j}| \geq \varepsilon \right) \leq 2 \exp\left( -\frac{m \varepsilon^2}{2 \sigma^2_{\nu j}} \right),
\]

where the variance of $\beta_{\nu j}$ is $\sigma^2_{\nu j} m^{-1}$. Assuming that \(\hat{i}_{\nu \nu} (\hat{\beta}_{\nu j})\) is a consistent estimator of \(i_{\nu \nu} (\beta_{\nu j})\),  it follows that \citep{van2000asymptotic}
\[
\mathbb{P}\left( |\hat{i}_{\nu \nu} (\hat{\beta}_{\nu j}) - i_{\nu \nu} ({\beta}_{\nu j})| \geq \delta \right) \leq 2 \exp\left( -C' m \delta^2 \right),
\]

for some constant \(C' > 0\). Set \(\varepsilon\) and \(\delta\) such that $
2 |\beta_{\nu j} i_{\nu \nu} ({\beta}_{\nu j})| \varepsilon + |\beta_{\nu j}^2| \delta = \lambda_{\theta} m^{-\theta}.
$

For simplicity, choose $ \varepsilon = \frac{\lambda_{\theta} m^{-\theta}}{4 |\beta_{\nu j} i_{\nu \nu}|}, \quad \delta = \frac{\lambda_{\theta} m^{-\theta}}{2 |\beta_{\nu j}^2|}.
$ Therefore
\[
\mathbb{P}\left( |\hat{\omega}_{\nu j} - \omega_{\nu j}| \geq \lambda_{\theta} m^{-\theta} \right) \leq \mathbb{P}\left( |\hat{\beta}_{\nu j} - \beta_{\nu j}| \geq \varepsilon \right) + \mathbb{P}\left( |\hat{i}_{\nu \nu} (\hat{\beta}_{\nu j}) - i_{\nu \nu} ({\beta}_{\nu j})| \geq \delta \right).
\]
Substituting the values of \(\varepsilon\) and \(\delta\) already defined, we can write
\[
\mathbb{P}\left( |\hat{\omega}_{\nu j} - \omega_{\nu j}| \geq \lambda_{\theta} m^{-\theta} \right) \leq 2 \exp\left( -\frac{m \varepsilon^2}{2 \sigma^2} \right) + 2 \exp\left( -C' m \delta^2 \right).
\]
We notice that
$\frac{m \varepsilon^2}{2 \sigma^2} = \frac{m}{2 \sigma^2} \left( \frac{\lambda_{\theta} m^{-\theta}}{4 |\beta_{\nu j} i_{\nu \nu} ({\beta}_{\nu j})|} \right)^2 = \frac{\lambda_{\theta}^2 m^{1 - 2\theta}}{32 \sigma^2 \beta_{\nu j}^2 i_{\nu \nu}^2 ({\beta}_{\nu j})} >0
$ and $
C' m \delta^2 = C' m \left( \frac{\lambda_{\theta} m^{-\theta}}{2 |\beta_{\nu j}^2|} \right)^2 = \frac{C' \lambda_{\theta}^2 m^{1 - 2\theta}}{4 \beta_{\nu j}^4}>0
$,  both arguments are non-negative. Since both exponents are proportional to \(m^{1 - 2\theta}\), for \(\theta \in [0, \frac{1}{2} )\), the exponents grow with \(m\). Therefore

\[
\mathbb{P}\left( |\hat{\omega}_{\nu j} - \omega_{\nu j}| \geq \lambda_{\theta} m^{-\theta} \right) \leq 4 \exp\left( - C_{\theta} m^{1 - 2\theta} \right),
\]

where the positive constant $
C = \min\left\{ \frac{\lambda_{\theta}^2}{32 \sigma^2 \beta_{\nu j}^2 i_{\nu \nu}^2}, \frac{C' \lambda_{\theta}^2}{4 \beta_{\nu j}^4} \right\}.
$
\end{proof}

\begin{theorem}
\label{th:ConvProfDeltaMethod}
Let $\boldsymbol{\phi} : \mathcal{D}_{\boldsymbol{\phi}} \subset \mathbb{R}^{W}\to \mathbb{R}$ be a map defined on a subset of $\mathbb{R}^W$, differentiable at $\boldsymbol{\delta}$.
Let $\hat{\boldsymbol{\delta}}$ be a random vector taking their values in the domain of $\boldsymbol{\phi}(\boldsymbol{\delta})=\boldsymbol{\delta} \boldsymbol{\mathcal{I}} ({\boldsymbol{\delta}})\boldsymbol{\delta}^{\top}$, where $\boldsymbol{\delta}$
 being a vector of dimension $W$ and $ \boldsymbol{\mathcal{I}} ({\boldsymbol{\delta}})$
 being a positive definite symmetric matrix of size $W\times W$. If $\hat{\boldsymbol{\delta}}=\boldsymbol{\delta}+ O_p(1/\sqrt{n})$ and $\mathbf{S}=o(\sqrt{n})$ where $\mathbf{S}=\text{diag}{(\mathbf{S}_1, \mathbf{S}_2, \mathbf{S}_3)}$ and $\mathbf{S}_1, \mathbf{S}_2, \mathbf{S}_3$ are overall penalties which contain $\boldsymbol{\lambda}_1, \boldsymbol{\lambda}_2, \boldsymbol{\lambda}_3$, for any $\epsilon > 0$.
$$
\boldsymbol{\hat{\delta}} \boldsymbol{\mathcal{I}} ({\boldsymbol{\hat{\delta}}})\boldsymbol{\hat{\delta}}^{\top} \overset{p}{\to}\boldsymbol{\delta} \boldsymbol{\mathcal{I}} ({\boldsymbol{\delta}})\boldsymbol{\delta}^{\top}.
$$
\end{theorem}

\begin{proof}   
Let $\boldsymbol{\phi}(\boldsymbol{\delta})=\boldsymbol{\delta} \boldsymbol{\mathcal{I}} ({\boldsymbol{\delta}})\boldsymbol{\delta}^{\top}$ be a differentiable map at $\boldsymbol{\delta}$, since $\boldsymbol{\mathcal{I}} ({\boldsymbol{\delta}})$ is  continuous at the parameters. We have
\begin{align}
\label{eq:approxphi}
\boldsymbol{\phi}(\hat{\boldsymbol{\delta}})= \boldsymbol{\phi}({\boldsymbol{\delta}})+ \boldsymbol{\phi}'({\boldsymbol{\delta}}) (\hat{\boldsymbol{\delta}}- \boldsymbol{\delta})+O(||\hat{\boldsymbol{\delta}} - \boldsymbol{\delta}||^2),
\end{align}
where $\boldsymbol{\phi}'({\boldsymbol{\delta}})=\nabla (\boldsymbol{\delta} \boldsymbol{\mathcal{I}} ({\boldsymbol{\delta}})\boldsymbol{\delta}^{\top})$ denotes the gradient of the function evaluated at 
$\boldsymbol{\delta}$. Recalling that $\boldsymbol{H}_{p}=$ $\boldsymbol{H}-\mathbf{S}$ and $\boldsymbol{S}=o(\sqrt{n})$ (for a detailed discussion of the results for $\boldsymbol{S}$ and $\boldsymbol{H}$, we refer the reader to Supplementary Material B in \cite{Marra2020}). We can write 
\[
\boldsymbol{\phi}(\hat{\boldsymbol{\delta}})-\boldsymbol{\phi}({\boldsymbol{\delta}})\overset{p}{\to}\boldsymbol{\phi}'({\boldsymbol{\delta}}) (\hat{\boldsymbol{\delta}}- \boldsymbol{\delta}),
\]
Taking the absolute value of $\boldsymbol{\phi}(\hat{\boldsymbol{\delta}})-\boldsymbol{\phi}(\boldsymbol{\delta})$ from (\ref{eq:approxphi}), we have that $|\boldsymbol{\phi}(\hat{\boldsymbol{\delta}})-\boldsymbol{\phi}({\boldsymbol{\delta}})|\le |\boldsymbol{\phi}'({\boldsymbol{\delta}}) (\hat{\boldsymbol{\delta}}- \boldsymbol{\delta})|$, and since $\hat{\boldsymbol{\delta}}=\boldsymbol{\delta}+ O_p(1/\sqrt{n})$ we can write $||\hat{\boldsymbol{\delta}} - \boldsymbol{\delta}|| = O_p(1/\sqrt{n})$. Thus, we have:

\begin{align*}
|\boldsymbol{\phi}'({\boldsymbol{\delta}}) (\hat{\boldsymbol{\delta}}- \boldsymbol{\delta})| 
&= |\boldsymbol{\phi}(\hat{\boldsymbol{\delta}}) - \boldsymbol{\phi}(\boldsymbol{\delta}) - O(||\hat{\boldsymbol{\delta}} - \boldsymbol{\delta}||^2)|\\
&= |\boldsymbol{\phi}(\hat{\boldsymbol{\delta}}) - \boldsymbol{\phi}(\boldsymbol{\delta})| + O(||\hat{\boldsymbol{\delta}} - \boldsymbol{\delta}||^2)\\
&= O_p(||\hat{\boldsymbol{\delta}} - \boldsymbol{\delta}||) + O(||\hat{\boldsymbol{\delta}} - \boldsymbol{\delta}||^2)\\
&= O_p(1/\sqrt{n}) + O(1/n)\\
&= O_p(1/n).
\end{align*}

The first line comes from the Taylor expansion in (\ref{eq:approxphi}) by simply rearranging the terms. Then we notice that the difference between the functions can be approximated by the true difference plus an error term of the order $||\hat{\boldsymbol{\delta}} - \boldsymbol{\delta}||^2$. Finally, noting that as $n$ becomes large, the term $O_p(1/\sqrt{n})$ dominates  $O(1/n)$.

We can conclude that  $|\boldsymbol{\phi}'({\boldsymbol{\delta}}) (\hat{\boldsymbol{\delta}}- \boldsymbol{\delta})|=O_p(1/{n})$. Therefore 
\[
\lim_{n\to \infty}\mathbb{P}\bigg(|\boldsymbol{\phi}(\hat{\boldsymbol{\delta}})-\boldsymbol{\phi}({\boldsymbol{\delta}})|>\epsilon\bigg)=0.
\]

\end{proof}

\begin{corollary}
\label{cor:FIM}
In accordance with what is discussed in the  Theorem \ref{th:ConvProfDeltaMethod}, it can be deduced that $\hat{\omega}_{\nu j} \overset{p}{\to} \omega_{\nu j}$
as $n \to \infty$.
\end{corollary}

Theorems \ref{th:upperbound} and \ref{th:ConvProfDeltaMethod} with Corollary \ref{cor:FIM} posit that as $n \to \infty$, $\hat{\omega}_{\nu j}$ is a consistent estimator, rendering the proposed measure appropriate for application in the $BRBVS$ domain. Furthermore, we state that the Fisher information measure has an exponential decay upper limit under normal conditions. 
This assertion is confirmed not solely through theoretical tools but also via empirical evidence, as delineated in the simulation study described in Section \textcolor{red}{4}.
Note that the result obtained in Theorem \ref{th:upperbound}
 is crucial, as it provides the basis for proving the algorithm's consistency. Establishing the consistency of the measure, together with the exponential decay bound, is key to distinguishing relevant information from noisy data in the BRBVS algorithm. This will be further explored in the following Corollaries, which generalize Proposition 1 and Theorem 1 in \cite{Baranowski2020} when considering two sets of variables.

\begin{corollary}
\label{th:firstcorollary}
Let $\mathcal{R}{\nu}$ be two variable rankings 
based on  ${\omega}_{\nu j} \overset{\Delta}{=} [\mathbf{e}^{(\nu)}_j \odot \boldsymbol{\beta}_\nu]^{\top} \boldsymbol{\mathcal{I}}_{\nu \nu}(\boldsymbol{\beta}_\nu) [\mathbf{e}^{(\nu)}_j \odot \boldsymbol{\beta}_\nu]$ for $j = 1, \dots, p$ and $\nu = 1, 2$, such that the rankings are a permutation of $\{1,\dots, p\}$ satisfying $\hat{\omega}_{\nu 1} >\dots> \hat{\omega}_{\nu p} $. Under the following conditions

\begin{enumerate}[label= {\em [A.\arabic*]}]
\item Let $\mathbf{Z}_1, \dots, \mathbf{Z}_n$ be independent observations. For some $\theta \in [0 , 1/2)$ and any $\lambda_{\theta} > 0$ we have 
    \[
 \mathbb{P}\bigg( |\hat{\omega}_{\nu j}(\mathbf{Z}_1, \dots, \mathbf{Z}_n)- \omega_{\nu j}| \ge \lambda_\theta n ^{-\theta}\bigg)
   \leq 4 \exp\left( - C_{\theta} n^{\gamma} \right), \]

where $\gamma= 1 - 2\theta >0$, $C_{\theta}>0$ and $n$ is the sample size.

\item The index sets of important variables are defined as $\mathcal{S}_\nu \subset \{1_\nu,\dots, p_\nu\}$ with $\mathcal{S}=\mathcal{S}_1 \cup \mathcal{S}_2 =\{1_1, 1_2, \dots, p_1, p_2\}$ does not depend on $n$ or $p$, and could be potentially an empty set. 

\item For any $a \in \mathcal{S}$, there exists $\mathcal{M}_a \subset \{1, \dots, p\} \setminus \mathcal{S}$, such that for any  $a \in \mathcal{M}_a$ the distribution of $\{\hat{\omega}_{\nu j}\}_{j \in \mathcal{M}_a, \nu =1,2}$ is exchangeable and $|\mathcal{M}_a| \to \infty$ as $n \to \infty$.
\item There exists $\vartheta' \in (0, \theta ]$ where $\theta \in [0, 1/2]$,  $c_{\vartheta'}>0$ such that
\[
( \min_{j \in \mathcal{S}_1 \cup \mathcal{S}_2} \omega_{\nu j}-  \max_{j \notin \mathcal{S}_1 \cup \mathcal{S}_2 } \omega_{\nu j}) \ge c_{\vartheta'} n^{-\vartheta'}.
\]
\item The number of covariates $p \le C_1 \exp{(n^{b_1})}$ where $0 \le b_1 < \gamma$ with  $\gamma = 1-2\theta$.
\end{enumerate}

The unique top-ranked sets exist and are equal to $\mathcal{S} = \mathcal{S}_1 \cup \mathcal{S}_2$. 

\end{corollary}

\begin{proof}

The proof of this Corollary involve three parts 
\begin{enumerate}[label={ [Pt.\arabic*]}]
    \item Show that \( \pi(\mathcal{S}) = \mathbb{P}(\{R_{11}, \dots, R_{1 s_1}\} = \mathcal{S}_1 \cap \{R_{21}, \dots, R_{2 s_2}\} = \mathcal{S}_2 ) \to 1 \) as \(n \to \infty\).
    \item Consider a ranking \(\mathcal{A}_{s+1}\) where \(\mathcal{A}_{s+1} \setminus \mathcal{S} = \{a\}\), and show that \( \pi(\mathcal{A}_{s+1}) \to 0 \) as \(n \to \infty\).
    \item For each \(k' = 1, \dots, s-1\), demonstrate that there exists some \(\mathcal{A}_{k'}\) such that \(\lim \sup_{n \to \infty} {\pi (\mathcal{A}_{k'})>0}\).
\end{enumerate} 

Before showing what stated in the Corollary, it may be useful to discuss the assumptions [A.1]-[A.5].

\begin{enumerate}[label= { [A.\arabic*]}]
\item asserts that the estimator of the measure is consistent, which is crucial for ensuring that the sets $\mathcal{S}_\nu$, for $\nu = 1,2$, are top-ranked.

\item discusses how $|\mathcal{S}|$ is bounded by $n$, and combined with a diverging number of covariates $p$, this indicates that the number of important covariates is small. Assumptions 2 to 4, taken together, can be seen as the well-known ``sparsity'' assumption.

\item can be associated with the sparsity assumption, where it is required that only a few covariates have a significant impact on the time-to-events.

\item assumes that there is a gap separating the important variables from the remaining ones. This gap is allowed to decrease slowly to zero. Together, Assumptions 1 and 4 imply that the ranking based on the Fisher information measure possesses the sure independent screening property \citep{fan2008}.

\item restricts the maximum number of covariates but permits the ultra-high dimensional setting, where the number of covariates grows exponentially with $n^{b_1}$ for some $b_1 > 0$.

\end{enumerate}
Given these assumptions, the three parts [Pt.1], [Pt.2] and [Pt.3] are individually shown.

[Pt.1]
We aim to show that the probability 
\[
\pi(\mathcal{S}) = \mathbb{P}\left(\{R_{11}, \dots, R_{1 s_1}\} = \mathcal{S}_1 \cap \{R_{21}, \dots, R_{2 s_2}\} = \mathcal{S}_2 \right) \to 1 \text{ as } n \to \infty,
\]
i.e., the bivariate variable ranking allows us to identify two sets of relevant variables for each of the time-to-events. Define the event 
\[
\Xi = \left\{\min_{j \in \mathcal{S}_1} \hat{\omega}_{1j} > \max_{j \notin \mathcal{S}_1} \hat{\omega}_{1j} \cap \min_{j \in \mathcal{S}_2} \hat{\omega}_{2j} > \max_{j \notin \mathcal{S}_2} \hat{\omega}_{2j} \right\},
\]
which can be rewritten as
\[
\Xi = \left\{\min_{j \in \mathcal{S}_1 \cup \mathcal{S}_2} \hat{\omega}_{\nu j} >  \max_{j \notin \mathcal{S}_1 \cup \mathcal{S}_2} \hat{\omega}_{\nu j}\right\}.
\]
Note that if there are no ties in the rankings, the event \(\Xi\) is equivalent to the event \(\{R_{11}, \dots, R_{1 s_1}\} = \mathcal{S}_1 \cap \{R_{21}, \dots, R_{2 s_2}\} = \mathcal{S}_2\). Using Ass. A4, we have
\[
\pi(\mathcal{S}) \geq \mathbb{P}(\Xi) \geq \mathbb{P}\left(\max_{\mathcal{S}_1 \cup \mathcal{S}_2}|\hat{\omega}_{\nu j} - \omega_{\nu j}| < \epsilon\right),
\]
where \(\epsilon = \lambda_{\theta} m^{-\theta}/2\). Recalling the Bonferroni inequality, we can bound the probability above as follows. It is well known that
\[
\mathbb{P}\left(\bigcup_{\nu=1}^{2}\bigcup_{j=1}^{p} \{ |\hat{\omega}_{\nu j} - \omega_{\nu j}| > \epsilon\}\right) \leq \sum_{\nu=1}^{2} \sum_{j=1}^{p} \mathbb{P}(|\hat{\omega}_{\nu j} - \omega_{\nu j}| \geq \epsilon),
\]
taking the complement, we obtain
\[
\mathbb{P}\left(\bigcap_{\nu=1}^{2}\bigcap_{j=1}^{p} \{ |\hat{\omega}_{\nu j} - \omega_{\nu j}| \leq \epsilon\}\right) \geq 1 - \sum_{\nu=1}^{2} \sum_{j=1}^{p} \mathbb{P}(|\hat{\omega}_{\nu j} - \omega_{\nu j}| \geq \epsilon).
\]
Therefore, we have
\[
\mathbb{P}\left(\max_{\mathcal{S}_1 \cup \mathcal{S}_2} |\hat{\omega}_{\nu j} - \omega_{\nu j}| \leq \epsilon\right) \geq 1 - 2p \sup_{\mathcal{S}_1 \cup \mathcal{S}_2} \mathbb{P}\left(|\hat{\omega}_{\nu j} - \omega_{\nu j}| \geq \epsilon\right).
\]

Combining this result with those in Theorem \ref{th:upperbound}, we note that the right-hand side term is of order \(O(-n^\gamma)\). This is sufficient to assert that \(\mathcal{S}_\nu\) are top-ranked sets. It should be noted that compared with the algorithm presented in \cite{Baranowski2020}, the component involving the \(\nu\) sets also comes into play. However, the magnitude on $\nu$ is much smaller than \(p\), and hence asymptotically negligible. 
Note that such reasoning may no longer apply when the algorithm is generalized to an indefinite number $\nu$. This generalization is left for future research. 

\vspace{5pt}

[Pt.2]
Let us consider a ranking $\mathcal{A}_{s+1}=\{ \{1_1, 1_2, \dots, p_1, p_2\}: |\mathcal{A}_{s+1}|=s+1\}$ such that \(\mathcal{A}_{s+1} \setminus \mathcal{S} = \{a\}\). This means that the ranking \(\mathcal{A}_{s+1}\) contains the same elements as \(\mathcal{S}\), except for the element \(a\), which is ranked at the bottom of \(\mathcal{A}_{s+1}\) with associated $\hat{\omega}_a$. We want to prove that \(\pi(\mathcal{A}_{s+1}) \to 0\) as \(n \to \infty\).

We suppose that either \(\mathcal{S}_1\) or \(\mathcal{S}_2\) has no ties, and therefore \(\mathcal{S} = \mathcal{S}_1 \cup \mathcal{S}_2\) also has no ties. The no ties assumption in $\mathcal{S}$ is based on the idea that it is always possible to rank elements in different positions using the proposed measure \(\omega_{\nu j}\).
Noting that $\mathbb{P} \bigg( \{ \min_{j \in \mathcal{A}_{s+1}} \hat{\omega}_{\nu j} >  \max_{j \notin \mathcal{A}_{s+1}} \hat{\omega}_{\nu j} \} \cap \{ \min_{j \in \mathcal{S}} \hat{\omega}_{\nu j} >  \max_{j \notin \mathcal{S}} \hat{\omega}_{\nu j} \} \bigg) $ can be rewritten as
$ \mathbb{P} \bigg(\{ \hat{\omega}_a > \max_{j \notin \mathcal{A}_{s+1}} \hat{\omega}_{\nu j} \} \cap  \Xi\bigg),
$ we can write
\[
\mathbb{P}\bigg( \hat{\omega}_a >  \max_{j \notin \mathcal{A}_{s+1}} \hat{\omega}_{\nu j} \bigg) \leq \mathbb{P}\bigg( \hat{\omega}_a > \max_{j \in \mathcal{M}_a \setminus \{a\}} \hat{\omega}_{\nu j} \bigg).
\]
This last inequality simply denotes that the probability of being excluded from the top-ranked set \(\mathcal{S}\) is greater than the probability of being excluded from a generic set. We also assume that all the elements of \(\mathcal{M}_a\) have an exchangeable distribution. This implies that the variables not contained in \(\mathcal{S}\) are noisy, and as such, the distribution of the associated \(\omega_{\nu j}\) is indistinguishable. Considering that the proposed measure is based on Fisher information, this assumption is justified. We then write 
\[
\mathbb{P}\left(\hat{\omega}_{\nu j^*} > \max_{j^* \in \mathcal{M}_a \setminus \{a\}} \hat{\omega}_{\nu j^*}\right)
\]
and note that for each \(j^*\), we have the same probability (exchangeability). Thus, summing over the probability and using the exchangeability,

\[
\mathbb{P}\left(\hat{\omega}_a > \max_{j \in \mathcal{M}_a \setminus \{a\}} \hat{\omega}_{\nu j}\right) \leq \frac{1}{|\mathcal{M}_a|}
\]

where \( |\mathcal{M}_a| \) denotes the cardinality of \(\mathcal{M}_a \setminus \{a\}\).  Then, as \(n \to \infty\), $\frac{1}{|\mathcal{M}_a|} \to 0$. Therefore
\[
\pi_n(\mathcal{A}_{s+1}) \leq \mathbb{P} \left( \hat{\omega}_a > \max_{j \notin \mathcal{A}_{s+1}} \hat{\omega}_{\nu j} \right) + \mathbb{P}(\Xi^c) \leq \mathbb{P} \left( \hat{\omega}_a > \max_{j \in \mathcal{M}_a \setminus \{a\}} \hat{\omega}_{\nu j} \right) + \mathbb{P}(\Xi^c) \to 0.
\]

\vspace{5pt}
[Pt.3]
For every $k'=1,\dots, s-1$, we show that there exists some $\mathcal{A}_{k'}$ such that $\lim \sup_{n \to \infty} \pi(\mathcal{A}_{k'})>0$. If we sum over the all possible $k'$ such that $\mathcal{A}_{k'} \subset \mathcal{S}$ we have 

\[
\sum_{\{\mathcal{A}_{k'}: \mathcal{A}_{k'} \subset \mathcal{A} \}} \pi (\mathcal{A}_{k'}) \ge \mathbb{P} \bigg(\{ \min_{j \in \mathcal{S}_1 \cup \mathcal{S}_2} \hat{\omega}_{\nu j} >  \max_{j \notin \mathcal{S}_1 \cup \mathcal{S}_2} \hat{\omega}_{\nu j} \}\bigg) \to 0 \text{ as } n\to \infty
\]

this holds for what has been discussed in Pt2. However, there are $\binom{s}{k'}$ possible combinations in $\{\mathcal{A}_{k'}: \mathcal{A}_{k'} \subset \mathcal{A} \}$. Therefore 

\[
\max_{\{\mathcal{A}_k': \mathcal{A}_k' \subset \mathcal{A} \}} \lim_{n \to \infty} \sum \pi(\mathcal{A}_k') \ge \frac{1}{\binom{s}{k'}}
\]

this implies that $\mathcal{S}$ is a top-ranked set, where the uniqueness of $\mathcal{S}$ follows from $\pi(\mathcal{S}) \to 1$ as $n \to \infty$ and $\sum_{\mathcal{A}_s} \pi(\mathcal{A}_s)=1$.

\end{proof}



The Corollary \ref{th:firstcorollary} is a generalization of Proposition 1 presented in \cite{Baranowski2020} to the case where there are two sets of relevant variables, and the Fisher information measure is used for rankings.

\begin{corollary} 
\label{th:Secondcorollary}
Under the following assumptions:

\begin{enumerate}[label={\em [B.\arabic*]}]
\item Let $\mathbf{Z}_1, \dots, \mathbf{Z}_m$ be independent observations. For some $\theta \in [0 , 1/2)$ and any $\lambda_{\theta} > 0$ we have 
    \[
 \mathbb{P}\bigg( |\hat{\omega}_{\nu j}(\mathbf{Z}_1, \dots, \mathbf{Z}_m)- \omega_{\nu j}| \ge \lambda_\theta m ^{-\theta}\bigg)
   \leq 4 \exp\left( - C_{\theta} m^{\gamma} \right), \]

where $\gamma= 1 - 2\theta >0$, $C_{\theta}>0$ and $m>0$ is a function of $n$.

\item There exists constants $C_1>0$, $b_1\in(0, \gamma)$, with $\gamma =1-2\theta$, such that $p\le C_1 \exp{(n^{b_1})}$.

\item The subsample size $m$ is defined as $m \sim n^{b_2}$, with $b_2 \in (0,1)$ and $\gamma b_2- b_1 >0$ with $\gamma=1-2\theta$ and $b_1 \in (0, \gamma)$.

\item  The index sets of important variables defined $\mathcal{S}_\nu \subset \{1_\nu,\dots, p_\nu\}$ with $\mathcal{S}=\mathcal{S}_1 \cup \mathcal{S}_2 =\{1_1, 1_2, \dots, p_1, p_2\}$ does not depend on $n$ or $p$, and could be potentially an empty set.  For any $a \in \mathcal{S}$, there exists $\mathcal{M}_a \subset \{1, \dots, p\} \setminus \mathcal{S}$, such that for any  $a \in \mathcal{M}_a$ the distribution of $\{\hat{\omega}_{\nu j}\}_{j \in \mathcal{M}_a}$ is exchangeable, $|\mathcal{M}_a| \to \infty$ as $n \to \infty$ and $\min_{a \notin \mathcal{S}} |\mathcal{M}_a| \ge C_3 n^{b_3}$ with $C_3 >0$ and $b_3/2< 1-b_2< b_3$ with $b_3 \in (0,1)$.

\item There exists $\vartheta' \in (0, \theta ]$ where $\theta$ is defined in Ass. B1 and $c_{\vartheta'}>0$ such that
\[
( \min_{j \in \mathcal{S}_1 \cup \mathcal{S}_2} \omega_{\nu j}-  \max_{j \notin \mathcal{S}_1 \cup \mathcal{S}_2 } \omega_{\nu j}) \ge c_{\vartheta'} m^{-\vartheta'},
\]
as specified above $m \sim n^{b_2}$ with $b_2 \in (0,1)$.

\item The number of random draws $B$ is bounded in $n$.

\item The maximum subset size $k_{\text{max}} \in [s, C_4 n^{b_4}]$ with $C_4 > 0$ and $b_4$ satisfying $b_3>b_4$, where $b_3/2< 1-b_2< b_3$ with $b_3 \in (0,1)$.
\end{enumerate}

Let $\hat{\mathcal{S}} = \hat{\mathcal{A}}_{\hat{s}}$, where $\hat{\mathcal{A}}_{\hat{s}}$ is given by \eqref{eq:max_pi}. Then for any $\tau \in (0,1]$, there exists a constants $\beta, C_\beta >0$ such that $$\mathbb{P}\bigg( \hat{\mathcal{S}} \ne \mathcal{S}\bigg) = o \bigg(\exp{(- C_{\alpha} n^{\alpha})} \bigg) \to 0$$

as $n \to \infty$.
\end{corollary}

Note that the corollary just stated, along with its proof, is a minor result of Theorem 1 from \cite{Baranowski2020}, formulated when considering the set of relevant variables as $\mathcal{S} = \mathcal{S}_1 \cup \mathcal{S}_2$.

\begin{proof}

For convenience, let's set 

\begin{itemize}
    \item $\hat{\omega}_{\nu j}= \hat{\omega}_{\nu j}(\mathbf{Z}_1,\dots, \mathbf{Z}_m)$,
    \item $ \pi(\mathcal{S})= \pi_{m} (\mathcal{S})$,
    \item $\Gamma = \max_{\mathcal{A} \not\subset \mathcal{S}_1 \cup \mathcal{S}_2 : |\mathcal{A}| \le k_{\text{max}}  } \pi_m(\mathcal{A})$.
\end{itemize}

Some key results were already demonstrated in Corollary \ref{cor:FIM}. Here, we aim to show that by using the Fisher information measure to select the relevant variables, we can achieve a clear separation between $\pi(\mathcal{S})$ and $\Gamma$. Specifically, this means that the probability of obtaining the set $\mathcal{S} = \mathcal{S}_1 \cup \mathcal{S}_2$ is distinct from the probability of obtaining a generic set $\mathcal{A}$, indicating that the two events are well separated.


By applying the same reasoning as in Corollary \ref{th:firstcorollary}, we can write:

\[
\pi(\mathcal{S})\ge \mathbb{P}\bigg(\max_{j=1,\dots, p}  | \hat{\omega}_{\nu j}- \omega_{\nu j}| \le \epsilon\bigg) \ge 1- C_{\epsilon} \exp(-m^{\gamma}),
\]
for some $C_{\epsilon} > 0$, noting that we are assuming $p \leq C_1 \exp{(n^{b_1})}$, and that $m \sim n^{b_2}$ with $\gamma b_2 > b_1$ and $b_1 < \gamma$, where $\gamma = 1 - 2\theta$, $\theta \in [0, 1/2)$, we finally note that the following chain of inequalities is derived by considering a fraction of the sample size, namely $m$.  From this, we can conclude that $\pi (\mathcal{S}) = 1 - O(\exp{(-n^{b_2 \gamma})}) \to 1$ as $n \to \infty$.

For every $\mathcal{A}_k$ such that $|\mathcal{A}_k| = k$ with $k \le k_{\text{max}}$, and for at least one element $a \in \mathcal{A}_k \setminus \mathcal{S}$, if there are no ties in the rankings, using a similar argument as in Corollary \ref{th:firstcorollary}, in combination with the assumption that the distribution of $\{\hat{\omega}_{\nu j}\}_{j \in \mathcal{M}_a, \nu = 1, 2}$ is exchangeable, with $\text{min}_{a \notin \mathcal{S}} |\mathcal{M}_a| \ge C_3 n^{b_3}$ for some $C_3 > 0$, and that the maximum subset size is $k_{\text{max}} \in [s, C_4 n^{b_4}]$ with $b_4 < b_3$, we have

\begin{align}
    \pi(\mathcal{A}) \le \mathbb{P}\bigg( \hat{\omega}_a > \max_{j \in \mathcal{M}_a \setminus \mathcal{A}} \hat{\omega}_{\nu j}\bigg) \le \frac{1}{|\mathcal{M}_a-k_{\text{max}}|} \le \frac{1}{|C_3 n^{b_3}-C_4 n^{b_4}|},
\end{align}

where the following chain of inequalities holds even if we assume ties in the ranking. Thus, we conclude that $\Gamma = O(n^{-b_3})$.

To recap, so far we have shown that $\Gamma = O(n^{-b_3})$ and $\pi (\mathcal{S}) = 1 - O(\exp{(-n^{b_2 \gamma})})$, indicating that the two probabilities exhibit different rates of convergence. From these two results, we can observe that for sufficiently large $n$, it is natural to write the following chain of inequalities
\[
\Gamma \le n^{(-b_3+ \Delta)} \le n^{(-b_3+ \Delta)/2} \le \pi(\mathcal{S}),
\]
where $\Delta = (b_3+b_2-1)/2>0$ by assumption. We then define the following events 
\begin{align*}
    \mathcal{C}_k &= \bigcap_{k=1}^{k_{\text{max}}} \bigg\{\ \Gamma < t_1 \bigg\} \\
   \mathcal{B}&=\bigg\{ \hat{\pi}(\mathcal{S})>\frac{1}{2}\bigg\}\\
   \mathcal{C}&=\mathcal{B} \cap \bigcap_{k=1}^{k_{\text{max}}} \mathcal{C}_k,
\end{align*}
where the ultimate goal is to prove that $\mathbb{P}(\mathcal{C}) \to 1$ as $n \to \infty$.  In other words, the separation between the two events may be demonstrated by proving the convergence to one of the probability of the event just defined. This also shows that, as $n \to \infty$, the probability of selecting relevant variables is higher than selecting any set of size $k \le k_{\text{max}}$ of the type $\mathcal{A} \not\subset \mathcal{S}$.

To prove $\mathbb{P} (\mathcal{C}) \to 1$ as $n \to \infty$ when $B>0$ for sufficiently large $n$, we can use the Lemma A1 in \cite{Baranowski2020}. After applying the Bonferroni inequality and noticing that $\pi(\mathcal{S})$ is obtained averaging $B$ copies obtained from a single bootstrap sample we have

\[
\mathbb{P}(\mathcal{B}^c) \le B \exp{(-C'n^{\gamma b_2(1-b_2)/2})}.
\]

To bound the second event, we use another result from \cite{Baranowski2020}, Lemma A.3, which allows us to state that 

\[
\mathbb{P}(\mathcal{C}^c_k)\le \exp{(-C'' n^{1-b_2-b_3/2})}
\]

where $C''$ is a positive constant. For sufficiently large $n$, the right-hand side of the inequality tends to zero, as $1 - b_2 - b_3/2 > 0$.  For further details on the proof of these bounds, we refer to \cite{Baranowski2020}, as in the case of two sets of relevant variables $\mathcal{S} = \mathcal{S}_1 \cup \mathcal{S}_2$, the theoretical steps are exactly the same. Putting everything together, we can conclude that 

\begin{align*}
\mathbb{P}(\mathcal{C}) &\ge 1- k_{\text{max}} \exp{\bigg(-C'' n^{1-b_2-b_3/2}\bigg)}- \exp{\bigg( - C' n^{\gamma b_2 (1-b_2)/2}\bigg)}\\
&\ge 1- C_4n^{b_4} \exp{\bigg(- C'' n^{1-b_2-b_3/2}\bigg)}-\exp{\bigg( - C' n^{\gamma b_2 (1-b_2)/2}\bigg)}\\
&\ge 1-\exp{\bigg( - C_{\alpha} n^{\alpha}\bigg)}, 
\end{align*}

for $\alpha \in (0,1)$ and $C_{\beta} >0$, for sufficiency large $n$ we can conclude that $\mathbb{P}(\mathcal{C})\to \infty$. Notice that from $1/2>t_1$, one concludes that $\hat{\mathcal{A}}_{\hat{s}}= \mathcal{S}$, hence showing that $\hat{s}=s$ and proving the consistency of the estimation procedure, in other words

$$\mathbb{P}\bigg( \hat{\mathcal{S}} \ne \mathcal{S}\bigg) = o \bigg(\exp{(- C_{\alpha} n^{\alpha})} \bigg) \to 0$$

as $n \to \infty$ this concludes the proof.

\end{proof}

The results presented in this section demonstrate that the proposed measure for the class of Copula Link-Based Survival Additive Models \citep{Marra2020} is consistent and aligns with the assumptions of the algorithm. Additionally, through a straightforward generalization of the results presented in \cite{Baranowski2020}, we have shown that our findings on the variable selection are equally consistent.

\clearpage


\section{Simulation study: further results} \label{Ch:Simulation}

The simulation results of the scenarios described in Section \ref{Ch:Main_simulation} of the Main Paper are here further discussed.

Additional comments and evaluations arise from Table \ref{tab:SimMIvsAbs1} to Table \ref{tab:SimMIvsAbs4}. For {\em Scenario A}, it can be noted that the metric $\hat\omega$ always selects the two relevant covariates $\boldsymbol{x}_1$ and $\boldsymbol{x}_2$ of the first margin (Table \ref{tab:SimMIvsAbs1} and \ref{tab:SimMIvsAbs2}) and the percentage of cases where only $\{\boldsymbol{x}_1, \boldsymbol{x}_2\}$ are chosen increases with $n$.\\
Also in the second margin, all three covariates are always selected when $p=100$, even if this percentage slightly decreases when $p=200$ (as expected).
These last results on the second margin are almost replicated when the $\hat\phi$ metric is considered. \\
For the more complex {\em Scenario B}, the results in Tables \ref{tab:SimMIvsAbs3} and \ref{tab:SimMIvsAbs4} show that with the $\hat\omega$ metric the covariates $\boldsymbol{x}_1$ and $\boldsymbol{x}_2$ are always selected for the first margin,
and its performance improves as $n$ increases and $p=100$ (Table \ref{tab:SimMIvsAbs3}).
For the second margin, the three covariates  $\boldsymbol{x}_1$, $\boldsymbol{x}_2$ and $\boldsymbol{x}_3$ are always chosen, and they are frequently the only variables included in $\hat s_2$. Similar good results are presented in Table \ref{tab:SimMIvsAbs4} with $p=200$.\\
If we evaluate the performance of the $\hat\phi$ competing metric in  Scenario B, we can further confirm its lower accuracy (with respect to $\hat\omega$) in both cases, when $p=100$ (Table \ref{tab:SimMIvsAbs3}) and $p=200$ (Table \ref{tab:SimMIvsAbs4}). \\
All previous results further highlight how the selection of the proper metric for a given statistical model can improve the performance of the BRBVS algorithm and consequently the identification of the relevant variables.

\clearpage

  \begin{sidewaysfigure}
        \centering
        \includegraphics[width=1\textwidth]{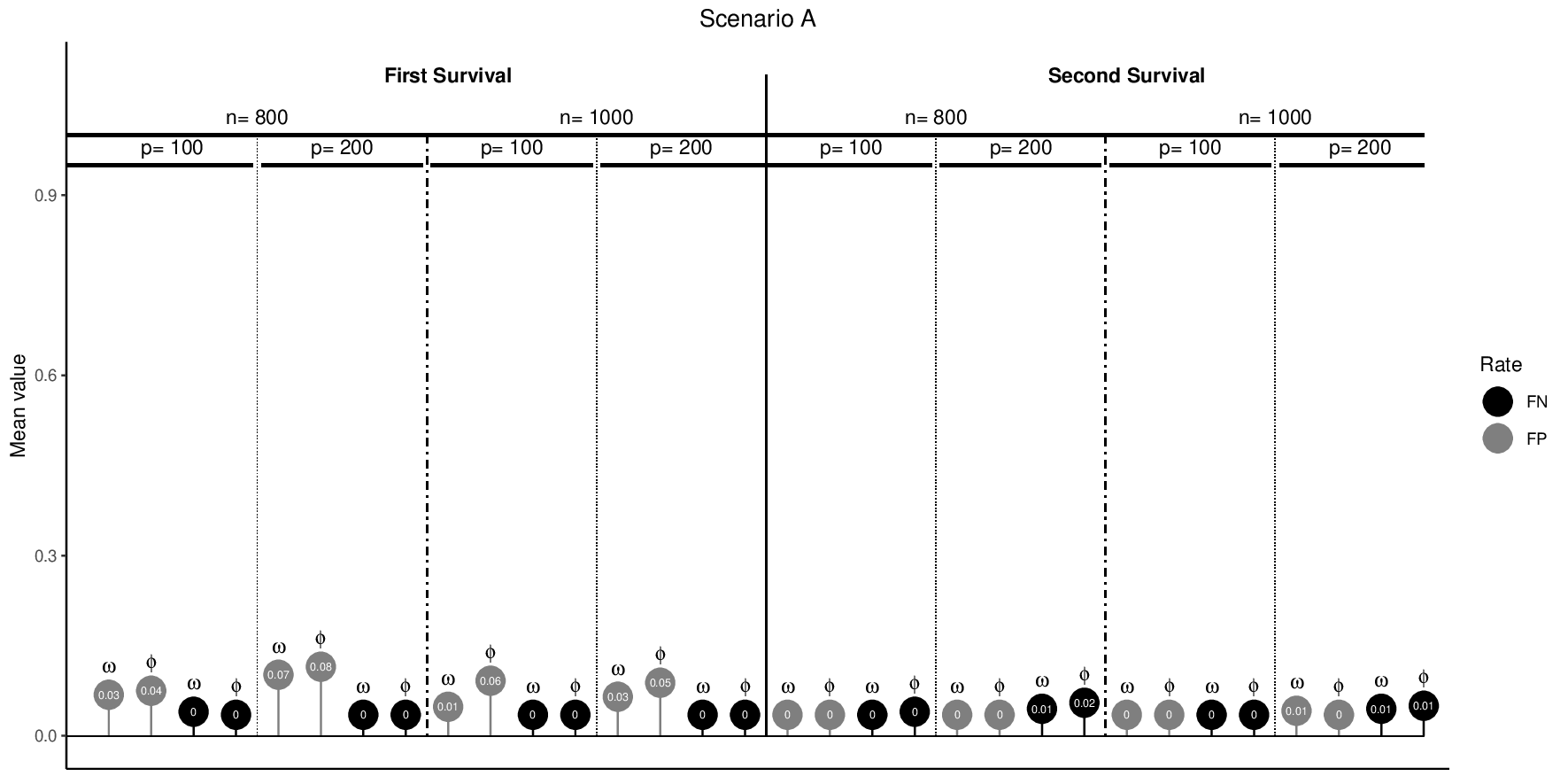}
        \caption{Simulation Results for Scenario A. Mean of False Negatives ($FN_\nu$) and False Positives ($FP_\nu$), in dark gray and light gray respectively, over the $n_{rep}$ Monte Carlo replicates.
        }
        \label{fig:PlotScenarioA}
    \end{sidewaysfigure}

  \begin{sidewaysfigure}
        \centering
        \includegraphics[width=1\textwidth]{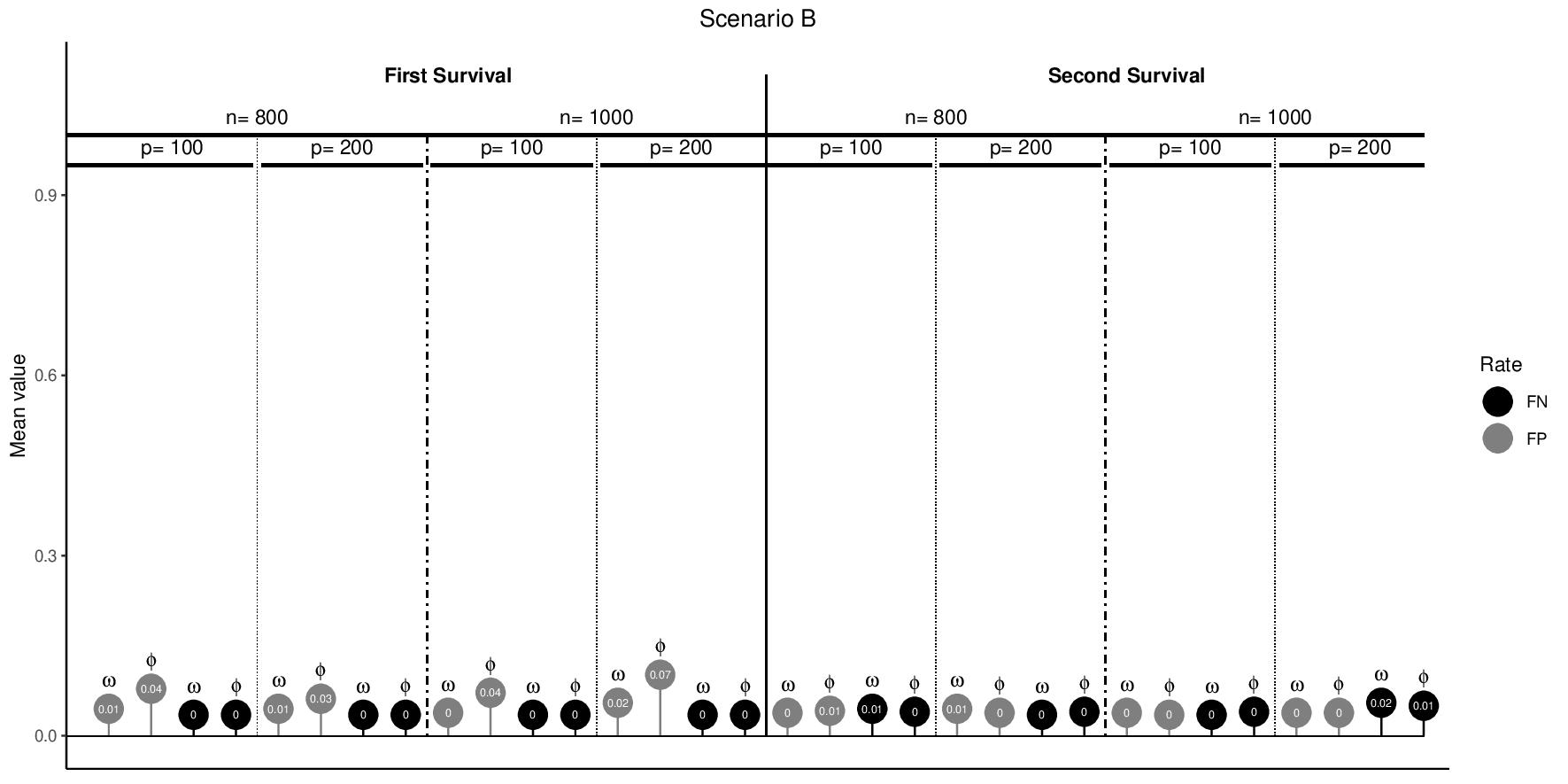}
        \caption{Simulation Results for Scenario B. Mean of False Negatives ($FN_\nu$) and False Positives ($FP_\nu$), in dark gray and light gray respectively, over the $n_{rep}$ Monte Carlo replicates.}
        \label{fig:PlotScenarioB}
    \end{sidewaysfigure}

\clearpage

\begin{table}[ht]

\begin{tabular}{lc|c}

   \multicolumn{3}{c}{$n=800;m=400; p=100$ }\\
  \multicolumn{3}{c}{$\text{\textbf{Scenario A}}: \eta_{3i}=\beta_{30}$ } \\ 
  \multicolumn{3}{c}{First Survival }\\
  \hline 
\hline
\multicolumn{1}{c}{\textbf{Sets }}&\multicolumn{1}{c|}{\textbf{Fisher Information Measure ($\hat\omega$)}} & \multicolumn{1}{c}{\textbf{Absolute value of the coefficients ($\hat\phi$)}} \\
 & \textbf{Freq. (\%)}  
 & \textbf{Freq. (\%)}  \\
 \hline
 \rowcolor{lgray}
  $\text{Sym}(\{\boldsymbol{x}_1$, $\boldsymbol{x}_2\})$ & 86   & 76   \\ 
  $\text{Sym}(\{\boldsymbol{x}_2$, $\boldsymbol{x}_3\})$ & 2   &0    \\ 
 $\text{Sym}(\{\boldsymbol{x}_1$, $\boldsymbol{x}_2$, $\boldsymbol{x}_3\})$ & 8   &24    \\
  $\text{Sym}(\{\boldsymbol{x}_1$, $\boldsymbol{x}_2, \mathbf{x}_j\})$ & 0   &0    \\ 
 $\text{Sym}(\{\boldsymbol{x}_1$, $\boldsymbol{x}_2$, $\boldsymbol{x}_3, \mathbf{x}_j\})$ & 4  &0    \\  

\hline

\multicolumn{3}{c}{ }\\
\multicolumn{3}{c}{Second Survival }\\
\hline
\hline

 & \textbf{Freq. (\%)}  
 & \textbf{Freq. (\%)}  \\
\hline
$\text{Sym}(\{\boldsymbol{x}_1$, $\boldsymbol{x}_3\})$ & 0 &   2   \\ 
\rowcolor{lgray}
 $\text{Sym}(\{\boldsymbol{x}_1$, $\boldsymbol{x}_2$, $\boldsymbol{x}_3\})$ & 100 &  98    \\     
   $\text{Sym}(\{\boldsymbol{x}_1$, $\boldsymbol{x}_2, \mathbf{x}_j\})$ & 0 &   0   \\ 
 $\text{Sym}(\{\boldsymbol{x}_1$, $\boldsymbol{x}_2$, $\boldsymbol{x}_3, \mathbf{x}_j\})$ & 0   &0    \\

 \hline

\end{tabular}
\caption{
Results from simulations highlight the frequency with which relevant sets are selected by the \texttt{BRBVS()} function in \texttt{R}. The notation $\text{Sym}(\cdot)$ signifies all possible permutations; for instance, $\text{Sym}(\{\boldsymbol{x}_1, \boldsymbol{x}_2\})$ includes both $\{\boldsymbol{x}_1, \boldsymbol{x}_2\}$ and $\{\boldsymbol{x}_2, \boldsymbol{x}_1\}$. The inclusion of $\mathbf{x}_j$ within a set indicates the selection of one or more \textit{noise} features by the algorithm. 
}
\label{tab:SimMIvsAbs1}
\end{table}

\begin{table}[ht]
\ContinuedFloat
\begin{tabular}{lc|c}

\multicolumn{3}{c}{ }\\
\multicolumn{3}{c}{ }\\
  \multicolumn{3}{c}{$n=1000;m=500; p=100$ }\\
  \multicolumn{3}{c}{$\text{\textbf{Scenario A}}: \eta_{3i}=\beta_{30}$ } \\ 
  \multicolumn{3}{c}{First Survival }\\
   \hline 
\hline
\multicolumn{1}{c}{\textbf{Sets }}&\multicolumn{1}{c|}{\textbf{Fisher Information Measure ($\hat\omega$)}} & \multicolumn{1}{c}{\textbf{Absolute value of the coefficients ($\hat\phi$)}} \\
 & \textbf{Freq. (\%)}  
 & \textbf{Freq. (\%)}  \\
 \hline
 \rowcolor{lgray}
  $\text{Sym}(\{\boldsymbol{x}_1$, $\boldsymbol{x}_2\})$ &92   &70    \\ 
 $\text{Sym}(\{\boldsymbol{x}_1$, $\boldsymbol{x}_2$, $\boldsymbol{x}_3\})$ &8    &28    \\ 
   $\text{Sym}(\{\boldsymbol{x}_1$, $\boldsymbol{x}_2, \mathbf{x}_j\})$ &0    &0    \\ 
 $\text{Sym}(\{\boldsymbol{x}_1$, $\boldsymbol{x}_2$, $\boldsymbol{x}_3, \mathbf{x}_j\})$ & 0   &2   \\  
\hline

\multicolumn{3}{c}{ }\\
\multicolumn{3}{c}{Second Survival }\\
\hline
\hline

 & \textbf{Freq. (\%)}  
 & \textbf{Freq. (\%)}  \\
\hline
 $\text{Sym}(\{\boldsymbol{x}_1$, $\boldsymbol{x}_3\})$ &  0&   0   \\ 
 \rowcolor{lgray}
 $\text{Sym}(\{\boldsymbol{x}_1$, $\boldsymbol{x}_2$, $\boldsymbol{x}_3\})$ & 100   &100    \\ 
  $\text{Sym}(\{\boldsymbol{x}_1$, $\boldsymbol{x}_2, \mathbf{x}_j\})$ & 0   &0    \\ 
 $\text{Sym}(\{\boldsymbol{x}_1$, $\boldsymbol{x}_2$, $\boldsymbol{x}_3, \mathbf{x}_j\})$ & 0   & 0   \\  
 \hline 
  
\end{tabular}
\caption*{Table \thetable: Cont.} 
\end{table}

\begin{table}[ht]

\begin{tabular}{lc|c}

   \multicolumn{3}{c}{$n=800;m=400;p=200$ }\\
  \multicolumn{3}{c}{$\text{\textbf{Scenario A}}: \eta_{3i}=\beta_{30}$ } \\ 
  \multicolumn{3}{c}{First Survival }\\
  \hline 
\hline
\multicolumn{1}{c}{\textbf{Sets }}&\multicolumn{1}{c|}{\textbf{Fisher Information Measure ($\hat\omega$)}} & \multicolumn{1}{c}{\textbf{Absolute value of the coefficients ($\hat\phi$)}} \\
 & \textbf{Freq. (\%)}  
 & \textbf{Freq. (\%)}  \\
 \hline
 \rowcolor{lgray} 
 $\text{Sym}(\{\boldsymbol{x}_1$, $\boldsymbol{x}_2\})$ & 74   &62    \\ 
 $\text{Sym}(\{\boldsymbol{x}_1$, $\boldsymbol{x}_2$, $\boldsymbol{x}_3\})$ & 14  & 30  \\ 
 $\text{Sym}(\{\boldsymbol{x}_1$, $\boldsymbol{x}_2, \mathbf{x}_j\})$ & 4  & 0   \\ 
 $\text{Sym}(\{\boldsymbol{x}_1$, $\boldsymbol{x}_2$, $\boldsymbol{x}_3, \mathbf{x}_j\})$ & 8  & 8   \\ 

\hline

\multicolumn{3}{c}{ }\\
\multicolumn{3}{c}{Second Survival }\\
\hline
\hline

 & \textbf{Freq. (\%)}  
 & \textbf{Freq. (\%)}  \\
\hline
 $\text{Sym}(\{\boldsymbol{x}_1$, $\boldsymbol{x}_3\})$ & 4   & 8    \\
 \rowcolor{lgray}
 $\text{Sym}(\{\boldsymbol{x}_1$, $\boldsymbol{x}_2$, $ \boldsymbol{x}_3\})$ & 96  & 92   \\ 
  $\text{Sym}(\{\boldsymbol{x}_1$, $\boldsymbol{x}_2, \mathbf{x}_j\})$ & 0   &0  \\ 
 $\text{Sym}(\{\boldsymbol{x}_1$, $\boldsymbol{x}_2$, $\boldsymbol{x}_3, \mathbf{x}_j\})$ &0   &0  \\ 
 \hline
\end{tabular}
\caption{Results from simulations highlight the frequency with which relevant sets are selected by the \texttt{BRBVS()} function in \texttt{R}. The notation $\text{Sym}(\cdot)$ signifies all possible permutations; for instance, $\text{Sym}(\{\boldsymbol{x}_1, \boldsymbol{x}_2\})$ includes both $\{\boldsymbol{x}_1, \boldsymbol{x}_2\}$ and $\{\boldsymbol{x}_2, \boldsymbol{x}_1\}$. The inclusion of $\mathbf{x}_j$ within a set indicates the selection of one or more \textit{noise} features by the algorithm. 
}
\label{tab:SimMIvsAbs2}
\end{table}

\begin{table}[ht]
\ContinuedFloat
\begin{tabular}{lc|c}
\multicolumn{3}{c}{ }\\
\multicolumn{3}{c}{ }\\
  \multicolumn{3}{c}{$n=1000;m=500;p=200$ }\\
  \multicolumn{3}{c}{$\text{\textbf{Scenario A}}: \eta_{3i}=\beta_{30}$ } \\ 
  \multicolumn{3}{c}{First Survival }\\
   \hline 
\hline
\multicolumn{1}{c}{\textbf{Sets }}&\multicolumn{1}{c|}{\textbf{Fisher Information Measure ($\hat\omega$)}} & \multicolumn{1}{c}{\textbf{Absolute value of the coefficients ($\hat\phi$)}} \\
 & \textbf{Freq. (\%)}  
 & \textbf{Freq. (\%)}  \\
 \hline
 \rowcolor{lgray}
 $\text{Sym}(\{\boldsymbol{x}_1$, $\boldsymbol{x}_2\})$ & 88   &70   \\ 
 $\text{Sym}(\{\boldsymbol{x}_1$, $\boldsymbol{x}_2$, $\boldsymbol{x}_3$\}) & 6   & 28   \\
 $\text{Sym}(\{\boldsymbol{x}_1$, $\boldsymbol{x}_2, \mathbf{x}_j\})$ & 0  & 0   \\ 
 $\text{Sym}(\{\boldsymbol{x}_1$, $\boldsymbol{x}_2$, $\boldsymbol{x}_3, \mathbf{x}_j\})$ & 6  & 2  \\ 
\hline

\multicolumn{3}{c}{ }\\
\multicolumn{3}{c}{Second Survival }\\
\hline
\hline

 & \textbf{Freq. (\%)}  
 & \textbf{Freq. (\%)}  \\
\hline
     $\text{Sym}(\{\boldsymbol{x}_1$, $\boldsymbol{x}_3\})$ & 4  &   6   \\ 
     \rowcolor{lgray}
 $\text{Sym}(\{\boldsymbol{x}_1$, $\boldsymbol{x}_2$, $\boldsymbol{x}_3\})$ & 94 &   94   \\ 
   $\text{Sym}(\{\boldsymbol{x}_1$, $\boldsymbol{x}_2, \mathbf{x}_j\})$ & 0   &0   \\ 
 $\text{Sym}(\{\boldsymbol{x}_1$, $\boldsymbol{x}_2$, $\boldsymbol{x}_3, \mathbf{x}_j\})$ & 2 &  0   \\  
 \hline 
  
\end{tabular}

\caption*{Table \thetable: Cont.}
\end{table}

\begin{table}[ht]

\begin{tabular}{lc|c}

   \multicolumn{3}{c}{$n=800;m=400;p=100$ }\\
   \multicolumn{3}{c}{$\text{\textbf{Scenario B}}: \eta_{3i}=\beta_{31}x_{1i}+\beta_{32}x_{2i}+\beta_{33}x_{3i} $} \\ 
  \multicolumn{3}{c}{First Survival }\\
  \hline 
\hline
\multicolumn{1}{c}{\textbf{Sets }}&\multicolumn{1}{c|}{\textbf{Fisher Information Measure ($\hat\omega$)}} & \multicolumn{1}{c}{\textbf{Absolute value of the coefficients ($\hat\phi$)}} \\
 & \textbf{Freq. (\%)}  
 & \textbf{Freq. (\%)}  \\
 \hline
 \rowcolor{lgray}
 $\text{Sym}(\{\boldsymbol{x}_1$, $\boldsymbol{x}_2\})$ & 94 &   74    \\ 
 $\text{Sym}(\{\boldsymbol{x}_1$, $\boldsymbol{x}_2$, $\boldsymbol{x}_3\})$ & 6  & 26   \\ 

  $\text{Sym}(\{\boldsymbol{x}_1$, $\boldsymbol{x}_2 ,\mathbf{x}_j\})$ &0  &0     \\ 
 $\text{Sym}(\{\boldsymbol{x}_1$, $\boldsymbol{x}_2$, $\boldsymbol{x}_3, \mathbf{x}_j\})$ &0   &0    \\ 

\hline

\multicolumn{3}{c}{ }\\
\multicolumn{3}{c}{Second Survival }\\
\hline
\hline

 & \textbf{Freq. (\%)}  
 & \textbf{Freq. (\%)}  \\
\hline
 $\text{Sym}(\{\boldsymbol{x}_1$, $\boldsymbol{x}_3\})$ &4   & 2  \\ 
 \rowcolor{lgray}
 $\text{Sym}(\{\boldsymbol{x}_1$, $\boldsymbol{x}_2$, $\boldsymbol{x}_3\})$ & 94   & 94   \\   
  $\text{Sym}(\{\boldsymbol{x}_1$, $\boldsymbol{x}_2, \mathbf{x}_j\})$ & 0   &0    \\ 
 $\text{Sym}(\{\boldsymbol{x}_1$, $\boldsymbol{x}_2$, $\boldsymbol{x}_3, \mathbf{x}_j\})$ &2    &4    \\

 \hline
\end{tabular}
\caption{
Results from simulations highlight the frequency with which relevant sets are selected by the \texttt{BRBVS()} function in \texttt{R}. The notation $\text{Sym}(\cdot)$ signifies all possible permutations; for instance, $\text{Sym}(\{\boldsymbol{x}_1, \boldsymbol{x}_2\})$ includes both $\{\boldsymbol{x}_1, \boldsymbol{x}_2\}$ and $\{\boldsymbol{x}_2, \boldsymbol{x}_1\}$. The inclusion of $\mathbf{x}_j$ within a set indicates the selection of one or more \textit{noise} features by the algorithm. 
}
\label{tab:SimMIvsAbs3}
\end{table}

\begin{table}[ht]
\ContinuedFloat 
\begin{tabular}{lc|c}

\multicolumn{3}{c}{ }\\
\multicolumn{3}{c}{ }\\
  \multicolumn{3}{c}{$n=1000;m=500;p=100$ }\\
  \multicolumn{3}{c}{$\text{\textbf{Scenario B}}: \eta_{3i}=\beta_{31}x_{1i}+\beta_{32}x_{2i}+\beta_{33}x_{3i} $} \\  
  \multicolumn{3}{c}{First Survival }\\
   \hline 
\hline
\multicolumn{1}{c}{\textbf{Sets }}&\multicolumn{1}{c|}{\textbf{Fisher Information Measure ($\hat\omega$)}} & \multicolumn{1}{c}{\textbf{Absolute value of the coefficients ($\hat\phi$)}} \\
 & \textbf{Freq. (\%)}  
 & \textbf{Freq. (\%)}  \\
 \hline
 \rowcolor{lgray}
 $\text{Sym}(\{\boldsymbol{x}_1$, $\boldsymbol{x}_2\})$ & 98  &  78   \\ 
 $\text{Sym}(\{\boldsymbol{x}_1$, $\boldsymbol{x}_2$, $\boldsymbol{x}_3\})$ & 2   &22    \\ 
  $\text{Sym}(\{\boldsymbol{x}_1$, $\boldsymbol{x}_2, \mathbf{x}_j\})$ & 0   & 0   \\ 
 $\text{Sym}(\{\boldsymbol{x}_1$, $\boldsymbol{x}_2$, $\boldsymbol{x}_3, \mathbf{x}_j\})$ &0    & 0   \\ 
\hline

\multicolumn{3}{c}{ }\\
\multicolumn{3}{c}{Second Survival }\\
\hline
\hline

 & \textbf{Freq. (\%)}  
 & \textbf{Freq. (\%)}  \\
\hline
  $\text{Sym}(\{\boldsymbol{x}_1$, $\boldsymbol{x}_3\})$ &0   & 2   \\ 
  \rowcolor{lgray}
 $\text{Sym}(\{\boldsymbol{x}_1$, $\boldsymbol{x}_2$, $\boldsymbol{x}_3\})$ &98   & 98   \\  
  $\text{Sym}(\{\boldsymbol{x}_1$, $\boldsymbol{x}_2, \mathbf{x}_j\})$ & 0 &   0   \\ 
 $\text{Sym}(\{\boldsymbol{x}_1$, $\boldsymbol{x}_2$, $\boldsymbol{x}_3, \mathbf{x}_j\})$ &2    & 0   \\ 
 \hline

\end{tabular}
\caption*{Table \thetable: Cont.} 
\end{table}

\begin{table}[ht]

\begin{tabular}{lc|c}

   \multicolumn{3}{c}{$n=800;m=400;p=200$ }\\
\multicolumn{3}{c}{$\text{\textbf{Scenario B}}: \eta_{3i}=\beta_{31}x_{1i}+\beta_{32}x_{2i}+\beta_{33}x_{3i} $} \\ 
  \multicolumn{3}{c}{First Survival }\\
  \hline 
\hline
\multicolumn{1}{c}{\textbf{Sets }}&\multicolumn{1}{c|}{\textbf{Fisher Information Measure ($\hat\omega$)}} & \multicolumn{1}{c}{\textbf{Absolute value of the coefficients ($\hat\phi$)}} \\
 & \textbf{Freq. (\%)}  
 & \textbf{Freq. (\%)}  \\
 \hline
 \rowcolor{lgray}
  $\text{Sym}(\{\boldsymbol{x}_1$, $\boldsymbol{x}_2\})$ &94    &84    \\ 
$\text{Sym}(\{\boldsymbol{x}_1$, $\boldsymbol{x}_2$, $\boldsymbol{x}_3\})$ &6   &16    \\
 $\text{Sym}(\{\boldsymbol{x}_1$, $\boldsymbol{x}_2, \mathbf{x}_j\})$ &0    &0    \\ 
 $\text{Sym}(\{\boldsymbol{x}_1$, $\boldsymbol{x}_2$, $\boldsymbol{x}_3, \mathbf{x}_j\})$ &0    &0    \\

\hline

\multicolumn{3}{c}{ }\\
\multicolumn{3}{c}{Second Survival }\\
\hline
\hline

 & \textbf{Freq. (\%)}  
 & \textbf{Freq. (\%)}  \\
\hline
$\text{Sym}(\{\boldsymbol{x}_1$, $\boldsymbol{x}_3\})$ & 0   & 2  \\ 
\rowcolor{lgray}
 $\text{Sym}(\{\boldsymbol{x}_1$, $\boldsymbol{x}_2$, $\boldsymbol{x}_3\})$ &96    &96    \\
  $\text{Sym}(\{\boldsymbol{x}_1$, $\boldsymbol{x}_2, \mathbf{x}_j\})$ & 0   &0    \\ 
 $\text{Sym}(\{\boldsymbol{x}_1$, $\boldsymbol{x}_2$, $\boldsymbol{x}_3, \mathbf{x}_j\})$ & 4   & 2   \\

 \hline

 \end{tabular}
\caption{
Results from simulations highlight the frequency with which relevant sets are selected by the \texttt{BRBVS()} function in \texttt{R}. The notation $\text{Sym}(\cdot)$ signifies all possible permutations; for instance, $\text{Sym}(\{\boldsymbol{x}_1, \boldsymbol{x}_2\})$ includes both $\{\boldsymbol{x}_1, \boldsymbol{x}_2\}$ and $\{\boldsymbol{x}_2, \boldsymbol{x}_1\}$. The inclusion of $\mathbf{x}_j$ within a set indicates the selection of one or more \textit{noise} features by the algorithm. 
}
\label{tab:SimMIvsAbs4}
 \end{table}

\begin{table}[ht]
\ContinuedFloat 
\begin{tabular}{lc|c}

\multicolumn{3}{c}{ }\\
\multicolumn{3}{c}{ }\\
  \multicolumn{3}{c}{$n=1000;m=500;p=200$ }\\
  \multicolumn{3}{c}{$\text{\textbf{Scenario B}}: \eta_{3i}=\beta_{31}x_{1i}+\beta_{32}x_{2i}+\beta_{33}x_{3i} $} \\ 
  \multicolumn{3}{c}{First Survival }\\
   \hline 
\hline
\multicolumn{1}{c}{\textbf{Sets }}&\multicolumn{1}{c|}{\textbf{Fisher Information Measure ($\hat\omega$)}} & \multicolumn{1}{c}{\textbf{Absolute value of the coefficients ($\hat\phi$)}} \\
 & \textbf{Freq. (\%)}  
 & \textbf{Freq. (\%)}  \\
 \hline
 \rowcolor{lgray}
 $\text{Sym}(\{\boldsymbol{x}_1$, $\boldsymbol{x}_2\})$ & 92   &66   \\ 
 $\text{Sym}(\{\boldsymbol{x}_1$, $\boldsymbol{x}_2$, $\boldsymbol{x}_3\})$ & 6   &30    \\ 
  $\text{Sym}(\{\boldsymbol{x}_1$, $\boldsymbol{x}_2, \mathbf{x}_j\})$ &2    & 0   \\ 
 $\text{Sym}(\{\boldsymbol{x}_1$, $\boldsymbol{x}_2$, $\boldsymbol{x}_3, \mathbf{x}_j\})$ & 0   & 4   \\ 
\hline

\multicolumn{3}{c}{ }\\
\multicolumn{3}{c}{Second Survival }\\
\hline
\hline

 & \textbf{Freq. (\%)}  
 & \textbf{Freq. (\%)}  \\
\hline
  $\text{Sym}(\{\boldsymbol{x}_1$, $\boldsymbol{x}_3\})$ & 8   &6    \\ 
  \rowcolor{lgray}
 $\text{Sym}(\{\boldsymbol{x}_1$, $\boldsymbol{x}_2$, $\boldsymbol{x}_3\})$ & 90    &92    \\
   $\text{Sym}(\{\boldsymbol{x}_1$, $\boldsymbol{x}_2, \mathbf{x}_j\})$ &0    & 0   \\ 
 $\text{Sym}(\{\boldsymbol{x}_1$, $\boldsymbol{x}_2$, $\boldsymbol{x}_3, \mathbf{x}_j\})$ & 2    & 2   \\ 
 \hline 
 \hline
  
\end{tabular}
\caption*{Table \thetable: Cont.}

\end{table}

\clearpage
\subsection{Mutual Information vs Fisher Information Measure}
\label{ch:simMIvsFIM}

To empirically assess the performance of Mutual Information (MI) within the BRBVS algorithm framework and compare it to the performance of the proposed measure \eqref{eq:FisherMeasure} based on the Fisher Information, a Monte Carlo simulation study is conducted. Specifically, the study considers \(n_{rep} = 100\) replicates, with \(n = 800\) units and $p = \{100, 200\}$ covariates. The matrices \(\boldsymbol{X}\) and \(\boldsymbol{Z}\) are generated according to Section \ref{Ch:Simulation}.

Once the data are generated and both noisy and informative covariates are merged together with the margins and survival functions, the rank of each covariate with the two aforementioned measures is evaluated.

For the dependence margin, we consider the simpler setup of \textbf{Scenario A}: $\eta_{3i}=\beta_{30}=1.2$. We set the proportion of right censoring at 11\% for the first time-to-event and 32\% for the second. As previously mentioned, our analysis involves two distinct measures:
\begin{itemize}
    \item[(a)] the measure based on Fisher Information, given by \eqref{eq:FisherMeasure}; 
    \item[(b)] the Mutual Information, denoted as \( MI_{\nu j}(T_\nu, X_j)=\int_{[0,1]^2}
	C(\mathbf{u})\log{C(\mathbf{u})} d \mathbf{u} \), where $ \mathbf{u}=\big[F_{\eta_\nu}(z), $ $ F_{X_j}(x)\big] $ is a vector of cumulative distribution functions with associated copula density \( c(\mathbf{u})=\frac{\partial^2 C(\mathbf{u})}{\partial F_{\eta_\nu}(z)  \partial  F_{X_j}(x)} \),
 for $j=1,\ldots, p$ and $\nu=1,2$.
\end{itemize}

The Mutual Information is estimated by using the \texttt{copent} package in \texttt{R}, as in \cite{ma2022copula}.

Subsequently, for each \( j \in \{1, \dots, p\}\), the relative frequency, over the $n_{rep}$ Monte Carlo replicates, with which the \( j^{\text{th}} \) covariate occurred in a given position is assessed.

In Table \ref{tab:SimMIvsFIM} the results are summarized evaluating the ranking assigned to the selected covariates of both margins.  It can be clearly noted the better performance of the $\hat\omega$ measure with respect to the measure based on the MI ($\widehat{ MI}$). For both margins, the percentage of the cases where the relevant variables are correctly selected is quite high when $\hat{\omega}$ is used. On the contrary, the MI measure discards a relevant variable in the first margin and further includes two irrelevant variables, even if with low frequency.
The presented results clearly show the weakness of the MI for variable selection in the presence of bivariate survival models, also under the simple model of Scenario A. 

Given the weakness of the MI measure in {\em Scenario A}, we do not show the results for the more complex {\em Scenario B}.

\begin{table}[ht]
\centering
\begin{tabular}{cccc||cccc}
\multicolumn{8}{c}{$n=800, p=100$}\\
\hline
\multicolumn{4}{|c||}{\textbf{Fisher Information Measure ($\hat\omega$)}} & \multicolumn{4}{c|}{\textbf{Mutual Information ($\widehat{MI}$)}} \\
\hline
\textbf{Position} & \textbf{Features} & \textbf{Freq. (\%)} & \textbf{Margin} & \textbf{Position} & \textbf{Features} & \textbf{Freq. (\%)} & \textbf{Margin} \\
\hline

$1^{\text{st}}$ & $\boldsymbol x_{2}$  & $98\%$  & First   & $1^{\text{st}}$ & $\boldsymbol x_{2}$  & $20\%$ & First   \\
$2^{\text{nd}}$ & $\boldsymbol x_{1}$  & $98\%$ & First   & $2^{\text{nd}}$ & $\boldsymbol{x}_{20}$ & $2\%$  & First   \\
   - &   -  &   -  &                 -            & $3^{\text{rd}}$ & $\boldsymbol x_{19}$ & $2\%$  & First   \\
    \hline
    \hline
$1^{\text{st}}$ & $\boldsymbol x_{3}$  & $72\%$  & Second  & $1^{\text{st}}$ & $\boldsymbol x_{1}$  & $50\%$ & Second  \\
$2^{\text{nd}}$ & $\boldsymbol x_{1}$  & $98\%$  & Second  & $2^{\text{nd}}$ & $\boldsymbol x_{3}$  & $54\%$ & Second  \\
$3^{\text{rd}}$ & $\boldsymbol x_2$  & $100\%$ & Second  & $3^{\text{rd}}$ & $\boldsymbol x_2$  & $64\%$ & Second  \\
\hline
\hline
\multicolumn{8}{c}{}\\
\multicolumn{8}{c}{$n=800, p=200$}\\
\hline
\multicolumn{4}{|c||}{\textbf{Fisher Information Measure ($\hat\omega$)}} & \multicolumn{4}{c|}{\textbf{Mutual Information ($\widehat{MI}$)}} \\
\hline
\textbf{Position} & \textbf{Features} & \textbf{Freq. (\%)} & \textbf{Margin} & \textbf{Position} & \textbf{Features} & \textbf{Freq. (\%)} & \textbf{Margin} \\
\hline

$1^{\text{st}}$ & $\boldsymbol x_2$  & $100\%$  & First   & $1^{\text{st}}$ & $\boldsymbol x_{182}$  & $6\%$ & First   \\
$2^{\text{nd}}$ & $\boldsymbol x_{1}$  & $98\%$ & First   & $2^{\text{nd}}$ & $\boldsymbol x_{2}$ & $2\%$  & First   \\
   - &   -  &   -  &                 -            & $3^{\text{rd}}$ & $\boldsymbol x_{1}$ & $2\%$  & First   \\
    \hline
    \hline
$1^{\text{st}}$ & $\boldsymbol x_{3}$  & $58\%$  & Second  & $1^{\text{st}}$ & $\boldsymbol x_{3}$  & $52\%$ & Second  \\
$2^{\text{nd}}$ & $\boldsymbol x_{1}$  & $90\%$  & Second  & $2^{\text{nd}}$ & $\boldsymbol x_{1}$  & $66\%$ & Second  \\
$3^{\text{rd}}$ & $\boldsymbol x_2$  & $100\%$ & Second  & $3^{\text{rd}}$ & $\boldsymbol x_2$  & $48\%$ & Second  \\
\hline
\end{tabular}
\caption{Monte Carlo Simulation results $p=\{100,200\}$. Comparison of the covariates ranking obtained from the two measures based on the Mutual Information ($\widehat{MI}$) and on the Fisher Information ($\hat \omega$) respectively. The relative frequency (in percentage) of  observing the \( j ^{\text{th}}\) covariate in the top three positions is evaluated for the two survivals; the covariates with percentage lower than \( 1\% \) are disregarded. 
 }
\label{tab:SimMIvsFIM}
\end{table}

\clearpage

\section{Bivariate Variable Selection via \texttt{BRBVS} \texttt{R}-package}
\label{ch:Rsoftware}
\subsection{The BRBVS Algorithm}
The Bivariate Ranking-Based variable Selection (BRBVS) method is implemented in a newly developed \texttt{R} package called \texttt{BRBVS}. This package offers both textual and visual outputs, enhancing usability for practitioners.

The BRBVS algorithm can be sketched in four main steps as in Algorithm \ref{alg:BRBVS}: the first two for the variable ranking and the last two for the variable selection.

\begin{algorithm}
 \caption{BRBVS Algorithm.}\label{alg:BRBVS}
 {{\bf Input:} $BRBVS(\boldsymbol Z, B, m, k_{max}, \tau$)}
\BlankLine
 \# $\boldsymbol Z$: $[n\times (p+2)]$ data matrix with $\mathbf{z}_i=(t_{i1}, t_{i2}, x_{i1}, \ldots, x_{ip})$;\\
 \# $B$: bootstrap replicates;\\
 \# $m$: subsets size;\\
 \# $k_{max}$: maximum number of covariates in each $\eta_\nu$, for $\nu=1,2$;\\
 \# $\tau$: fixed value to calculate the ratio \eqref{eq:hat_s_v}\\
 \For{$b\leftarrow 1$ \KwTo $B$}{
 {\bf Step 1:} draw (uniformly without replacement) $r$ subsamples from $\boldsymbol Z$, with $r=n/m$;\\
 {\bf Step 2:} estimate the weights $\hat\omega_{\nu j}(\boldsymbol Z)$ for $j=1,\ldots, p$ and $\nu=1, 2$ and rank the corresponding variables
}
{\bf Step 3}: estimate the probabilities $\hat\pi_{\nu, m}(\mathcal A_v)$ in \eqref{eq:hat_pi_v}, for $\nu=1,2$\\
{\bf Step 4:} given $\tau$, select the sets of important covariates $\hat s_\nu$, for $\nu=1,2$\\
{{\bf Output:} $(\hat s_1, \hat s_2)$}
\end{algorithm}

\begin{remark}
 Note that the algorithm requires to fix three parameters: $(B,\tau, k_{\max})$. The number of bootstrap replicates $B$ is defined to jointly combine the accuracy of the estimates and the control of the computational burden; the value assigned to $\tau$ is not evaluated critical for the algorithm in \cite{Baranowski2020}, mainly when $\tau$ is not too close to zero. For this reason it is here fixed to $0.5$ (following the cited literature). Finally, the value $k_{\max}$ needs to be selected, also trying to reduce the computational weight of the algorithm. We recommend to start with $k_{\max}<\min(n,p)$ and if the number of selected variables is exactly $k_{\max}$, this value needs to be increased until $\max(|\hat s_1|, |\hat s_2|)<k_{max}$. $\qed$
 \end{remark}

\subsection{The BRBVS Algorithm applied to the AREDS data}
The Age-Related Eye Disease Study (AREDS) is a comprehensive, multi-center, long-term study aimed at understanding the clinical progression of age-related macular degeneration (AMD) and cataract. The study encompasses both observational research and clinical trials evaluating the efficacy of high-dose vitamin and mineral supplements in treating AMD and cataracts. Participants aged between 55 to 80 years at the time of enrollment are selected based on their ocular condition, excluding those with any significant illness or condition affecting long-term participation or medication adherence. The original study successfully enrolled 4,757 participants across various stages of AMD, monitored through extensive visual and ophthalmological assessments. The clinical trials within AREDS, which last for a median of 6.5 years, are followed by an additional five-year period of natural history data collection \citep{AREDSstudy1999}.\\
The primary event of interest in this study is the late-stage progression of AMD, a leading cause of blindness in developed nations.

 The dataset is subject to various types of censoring due to the periodic nature of assessments conducted semi-annually for the initial six years and annually thereafter. Specifically, it exhibits right censoring in \(46\%\) and \(44\%\) of cases and interval censoring in \(53\%\) and \(55\%\) of cases for the first and second eye respectively, \(22\%\) of the data are affected by mixed censoring.

The variables included in the AREDS dataset are described in Table \ref{tab:AREDS}.

\begin{sidewaystable}
\centering
\footnotesize
\begin{tabular}{@{}lcl@{}}
  \toprule
  \textbf{Features} & \textbf{Allowed values} & \textbf{Description} \\ 
  \midrule
  \texttt{SevScale1E} & $4-9$ & Severity scale associated with the right eye \\
  \texttt{SevScale2E} & $4-9$ & Severity scale associated with the left eye\\  
  \texttt{ENROLLAGE} & $1-99$ & Age at baseline\\ 
  \texttt{rs2284665} & $0, 1, 2$ & SNP covariate highly associated with late-AMD progression\\ 
  \texttt{enroll\_day1E} & decimal & Start of follow-up in days right eye \\
  \texttt{AMD\_recurrence1E} & decimal & Time to recurrence or last follow-up in days right eye\\
  \texttt{enroll\_day2E} & decimal & Start of follow-up in days left eye\\
  \texttt{AMD\_recurrence2E} & decimal & Time to recurrence or last follow-up in days left eye\\
  \texttt{censoring\_status1E} & recurrence, no recurrence & Recurrence censoring variable right eye\\
  \texttt{censoring\_status2E} & recurrence, no recurrence & Recurrence censoring variable left eye\\
  \bottomrule
\end{tabular}
\caption{AREDS data.}
\label{tab:AREDS}
\end{sidewaystable}

Below is an example of a function call utilized in the data analysis in Section \ref{Ch:Application}.

\begin{verbatim}
install.packages('BRBVS')
library(BRBVS)
data(AREDS)
Y<- AREDS[,c('t11','t12', 't21', 't22', 'cens1', 'cens2', 'cens')]
X<- AREDS[,c(3,4,5, 9)]
X$SevScale1E <- scale(as.numeric( X$SevScale1E))
X$SevScale2E <- scale(as.numeric(X$SevScale1E))
X$GG<- ifelse(X$rs2284665==0,1,0)
X$GT<- ifelse(X$rs2284665==1,1,0)
X$TT<- ifelse(X$rs2284665==2,1,0)
X$ENROLLAGE <- scale( X$ENROLLAGE)


Bivrbvs<- BRBVS(y=Y, x=X, 
                kmax=10, # maximum number of covariates in the sets
                tau=0.5, # threshold parameter usually in (0,1]
                n.rep=100, # number of bootstrap 
                copula='PL', # copula function
                margins=c('PO','PO'), #margins
                m=nrow(AREDS)/2 , # subsample
                metric='FIM' # Measure: Fisher information Measure
                )

summary(Bivrbvs)

Sets of Relevant Covariates
================================

Metric: FIM 
kmax: 10
Copula: PL 
Margins: PO PO 

================================

Survival Function  1 :
  -  1nd: SevScale1E (53.00%)
  -  2rd: SevScale2E (93.00%)
  -  3rd: ENROLLAGE  (22.00%)

Survival Function  2 :
  -  1nd: SevScale1E (47.00%)
  -  2rd: SevScale2E (91.00%)

plotBRBVS(Bivrbvs)
\end{verbatim}

Before calling the function  \texttt{BRBVS}, it is important to separate the time variables $t_{11}, t_{12}, t_{21}, t_{22}$ and those related to censoring from the other covariates. In our example, the time and censoring variables were assigned to the \texttt{Y} object, while the covariates $\{ \texttt{SevScale1E}, \texttt{SevScale2E}, \texttt{GG},$ $ \texttt{GT}, \texttt{TT}, \texttt{ENROLLAGE}\}$ in \texttt{X}. 

Starting from the top, \texttt{kmax} defines the maximum number of covariates considered in the two sets, \texttt{tau} is a threshold parameter, and \texttt{n.rep} specifies the number of bootstrap replications. The choice of copula is made through the parameter \texttt{copula}, which in our example sets a Plackett copula, and \texttt{margins} determines the margins for the copula, proportional odds in the code above. Various copula configurations can be chosen, for details we refer to the  \texttt{GJRM} package. Finally, the \texttt{m} parameter establishes the subsample, and \texttt{metric} defines the measure to be used. Three different measures can be specified: \texttt{FIM} for the Fisher Information Measure, \texttt{Abs} for the absolute value of the coefficients, and \texttt{CE} for the Copula Entropy (Mutual Information). Applying the \texttt{summary()} function directly to an object of the \texttt{BRBVS} class returns a summary table with the selection frequencies in percentage. In detail,  the analysis setup is summarized, recalling the measure employed, the maximum number of covariates in the two sets, the specified copula, and the margins. For each of the time-to-events, the ranking is shown along with the corresponding selection frequency. The function \texttt{plotBRBVS()} generates a bar plot that represents the selection frequency of the covariates. More details about the \texttt{BRBVS()} function can be found in the \texttt{BRBVS} package in \texttt{R}.

Once we have the two sets of relevant variables using the function \texttt{BRBVS}, we can estimate the model using the \texttt{gjrm()} function in the \texttt{GJRM} package and obtain parameter estimates.

\begin{verbatim}
eq1 <- t11 ~ s(t11, bs = "mpi") + s(ENROLLAGE) + SevScale1E + SevScale2E
eq2 <- t21 ~ s(t21, bs = "mpi")  + SevScale1E + SevScale2E
eq3 <- ~ SevScale1E + SevScale2E
\end{verbatim}
 we are specifying \(\eta_1\) (\texttt{eq1}), \(\eta_2\) (\texttt{eq2}), and \(\eta_3\) (\texttt{eq3}).
\begin{verbatim}
f.list <- list(eq1, eq2, eq3)

out <- gjrm(f.list, data = AREDS, surv = TRUE,
            copula = "PL", margins = c("PO", "PO"),
            cens1 = cens1, cens2 = cens2, model = "B",
            upperBt1 = 't12', upperBt2 = 't22')   
\end{verbatim}
We can then check the convergence of the model using the function. \texttt{conv.check()}.
\begin{verbatim}
conv.check(out)

Largest absolute gradient value: 4.166472e-05
Observed information matrix is positive definite
Eigenvalue range: [0.008964266,164318.3]

Trust region iterations before smoothing parameter estimation: 71
Loops for smoothing parameter estimation: 8
Trust region iterations within smoothing loops: 18
Estimated overall probability range: 0.02390308 0.9999404
Estimated overall density range: 5.964254e-05 7.811673


\end{verbatim}
The model successfully converged, as evidenced by the number of iterations completed by the trust region algorithm. Additionally, the probability values obtained are coherent. Furthermore, the positively defined nature of the observed information matrix reaffirms the reliability of our results.
\begin{verbatim}
summary(out)

COPULA:   Plackett
MARGIN 1: survival with -logit link
MARGIN 2: survival with -logit link

EQUATION 1
Formula: t11 ~ s(t11, bs = "mpi") + s(ENROLLAGE) + SevScale1E + SevScale2E

Parametric coefficients:
            Estimate Std. Error z value Pr(>|z|)    
(Intercept) -18.3487     4.3956  -4.174 2.99e-05 ***
SevScale1E5   0.6875     0.2645   2.599 0.009355 ** 
SevScale1E6   0.8084     0.2552   3.168 0.001537 ** 
SevScale1E7   1.7219     0.2750   6.261 3.81e-10 ***
SevScale1E8   2.5884     0.3608   7.175 7.25e-13 ***
SevScale2E5   0.4063     0.2813   1.444 0.148612    
SevScale2E6   0.8768     0.2692   3.258 0.001123 ** 
SevScale2E7   1.0164     0.2900   3.505 0.000457 ***
SevScale2E8   1.4862     0.3403   4.367 1.26e-05 ***
---
Signif. codes:  0 ‘***’ 0.001 ‘**’ 0.01 ‘*’ 0.05 ‘.’ 0.1 ‘ ’ 1

Smooth components' approximate significance:
               edf Ref.df   Chi.sq p-value    
s(t11)       6.649  7.674 1836.475  <2e-16 ***
s(ENROLLAGE) 1.625  2.039    5.807  0.0548 .  
---
Signif. codes:  0 ‘***’ 0.001 ‘**’ 0.01 ‘*’ 0.05 ‘.’ 0.1 ‘ ’ 1


EQUATION 2
Formula: t21 ~ s(t21, bs = "mpi") + SevScale1E + SevScale2E

Parametric coefficients:
            Estimate Std. Error z value Pr(>|z|)    
(Intercept) -30.4468    10.5514  -2.886 0.003907 ** 
SevScale1E5   0.2684     0.2572   1.043 0.296794    
SevScale1E6   0.3454     0.2459   1.405 0.160059    
SevScale1E7   0.8530     0.2592   3.291 0.000999 ***
SevScale1E8   0.9947     0.3286   3.027 0.002468 ** 
SevScale2E5   0.8682     0.2796   3.105 0.001903 ** 
SevScale2E6   1.2294     0.2747   4.475 7.64e-06 ***
SevScale2E7   2.2515     0.2988   7.534 4.92e-14 ***
SevScale2E8   3.4251     0.3632   9.431  < 2e-16 ***
---
Signif. codes:  0 ‘***’ 0.001 ‘**’ 0.01 ‘*’ 0.05 ‘.’ 0.1 ‘ ’ 1

Smooth components' approximate significance:
         edf Ref.df Chi.sq p-value    
s(t21) 7.374  8.204   3808  <2e-16 ***
---
Signif. codes:  0 ‘***’ 0.001 ‘**’ 0.01 ‘*’ 0.05 ‘.’ 0.1 ‘ ’ 1


EQUATION 3
Link function for theta: log 
Formula: ~SevScale1E + SevScale2E

Parametric coefficients:
            Estimate Std. Error z value Pr(>|z|)    
(Intercept)   0.9658     0.5089   1.898 0.057702 .  
SevScale1E5  -0.2302     0.5196  -0.443 0.657677    
SevScale1E6  -0.6698     0.4964  -1.349 0.177259    
SevScale1E7  -1.1222     0.5165  -2.173 0.029812 *  
SevScale1E8  -1.0086     0.6035  -1.671 0.094662 .  
SevScale2E5   1.2973     0.5464   2.374 0.017580 *  
SevScale2E6   1.9442     0.5167   3.763 0.000168 ***
SevScale2E7   1.7939     0.5365   3.344 0.000826 ***
SevScale2E8   1.1878     0.6251   1.900 0.057421 .  
---
Signif. codes:  0 ‘***’ 0.001 ‘**’ 0.01 ‘*’ 0.05 ‘.’ 0.1 ‘ ’ 1

theta = 6.44(3.13,13.8)  tau = 0.354(0.203,0.496)
n = 628  total edf = 42.6

\end{verbatim}
Our analysis involves three distinct margin specifications, and correspondingly, the \texttt{summary()} function returns three tables, each dedicated to one margin. These output tables are organized into two key sections: `Parametric Effects' and `Smooth Effects'. In the `Parametric Effects' section, we present the estimates alongside their associated standard errors, including the z-value and the corresponding p-values. The `Smooth Effects' section focuses on the complexity of the model's smooth components. This is detailed by the number of degrees of freedom (\texttt{edf}), where, for instance, a value of `1' suggests a simple straight line, while `2' suggests a curve. Additionally, \texttt{Ref.df} (Reference degrees of freedom) and \texttt{chi-square} values are utilized to assess the statistical significance of these smooth components. Finally,  we have the value of \texttt{theta}, which represents $\theta_i= m\left\{\eta_{3 i}\left(\mathbf{x}_{3 i}; \boldsymbol{\beta}_{3}\right)\right\}$,  \texttt{tau} denotes the estimate of Kendall's tau, the sample size \texttt{n} and the estimate of the degrees of freedom  \texttt{total.edf} for the model specified. The functions \texttt{AIC()} and \texttt{BIC()} work in a similar fashion as those for the classical statistical models. More details about the model output produced by the function \texttt{gjrm()} can be found in the \texttt{GJRM} package in \texttt{R}.

\end{appendix}

\end{document}